\documentclass{article}
\usepackage[margin=1in]{geometry}
\usepackage[onehalfspacing]{setspace}
\usepackage{amsmath,mathtools}
\usepackage{pgfplots}
\usepackage{subcaption}
\usepackage{booktabs}
\usepackage[braket, qm]{qcircuit}
\usepackage{graphicx}
\usepackage{authblk}
\usepackage{xcolor}
\usepackage{booktabs}
\usepackage{subcaption}
\usepackage[linesnumbered,ruled,vlined]{algorithm2e}
\usepackage{qcircuit}

\SetCommentSty{mycommfont}

\SetKwInput{KwInput}{Input}                
\SetKwInput{KwOutput}{Output}              
\SetKwProg{Fn}{Function}{}{}

\usepackage[english]{babel}

\usepackage{pdflscape}
\usepackage{fancyvrb}
\usepackage{calc}
\usepackage{rotating}
\usepackage{pgf,tikz}
\usepackage{mathrsfs}
\usetikzlibrary{arrows}
\usepackage{mathtools}
\DeclarePairedDelimiter\ceil{\lceil}{\rceil}

\usepackage{listings}
\usepackage[]{qcircuit}
\usepackage{braket}
\usepackage{mathtools}
\usepackage{varwidth}
\usepackage{caption}
\usepackage{subcaption}
\usepackage{graphicx}
\usepackage[export]{adjustbox}

\usepackage{float}

\makeatletter
\def\paragraph{\@startsection{paragraph}{4}%
	\z@\z@{-\fontdimen2\font}%
	{\normalfont\bfseries}}
\makeatother

\usepackage{graphicx}
\usepackage{pgffor}
\usepackage{tikz,tikz-cd}
\usetikzlibrary{backgrounds,fit,decorations.pathreplacing}  
\usepackage{booktabs}
\usepackage{enumerate}

\usepackage{amssymb, amsmath, amsthm}

\usepackage{xypic}
\usepackage{scalerel}
\usepackage{ulem}

\usepackage{graphicx}              
\usepackage{amsmath}               
\usepackage{amsfonts}              
\usepackage{amsthm}                
\usepackage{xkeyval}               
\usepackage{amssymb}
\usepackage{enumerate}             
\usepackage{color}
\usepackage{breqn}                 
\usepackage{bm}                    
\usepackage{cite}

\allowdisplaybreaks[1]
\usepackage{nicefrac}
\usepackage{booktabs}

\usepackage{comment}

\usepackage{kpfonts}
\usepackage{stackengine}
\usepackage{calc}
\newlength\shlength
\newcommand\xshlongvec[2][0]{\setlength\shlength{#1pt}%
	\stackengine{-5.6pt}{$#2$}{\smash{$\kern\shlength%
			\stackengine{7.55pt}{$\mathchar"017E$}%
			{\rule{\widthof{$#2$}}{.57pt}\kern.4pt}{O}{r}{F}{F}{L}\kern-\shlength$}}%
	{O}{c}{F}{T}{S}}

\usepackage{hyperref}
\hypersetup{
	colorlinks   = true,    
	citecolor    = red      
}

\newcommand{\RN}[1]{%
	\textup{\uppercase\expandafter{\romannumeral#1}}%
}

\allowdisplaybreaks

\newcommand{\meqref}[1]{\text{Eq}.~\eqref{#1}}
\newcommand{\mref}[1]{Sec.~$ \!\ref{#1} $}
\newcommand{\mfig}[1]{Fig.~$ \!\ref{#1} $}

\newtheorem{thm}{Theorem}[subsection]

\newtheorem{lem}[thm]{Lemma}

\newtheorem{example}[thm]{Example}

\def\NN{{\mathbb N}}

\def\<{\langle}
\def\>{\rangle}

\numberwithin{equation}{section}

\newcommand{\abs}[1]{\lvert#1\rvert}

\usepackage{pgfplots}

\makeatletter
\def\smallunderbrace#1{\mathop{\vtop{\m@th\ialign{##\crcr
				$\hfil\displaystyle{#1}\hfil$\crcr
				\noalign{\kern3\p@\nointerlineskip}%
				\tiny\upbracefill\crcr\noalign{\kern3\p@}}}}\limits}
\makeatother

\begin{document}

	\title{An efficient quantum algorithm for preparation of uniform quantum superposition states}

	\author[1]{Alok Shukla \thanks{Corresponding author.}}
	\author[2]{Prakash Vedula}
	\affil[1]{School of Arts and Sciences, Ahmedabad University, India}
	\affil[1]{alok.shukla@ahduni.edu.in}
	\affil[2]{School of Aerospace and Mechanical Engineering, University of Oklahoma, USA}
	\affil[2]{pvedula@ou.edu}
	

	\maketitle


\begin{abstract}
Quantum state preparation involving a uniform superposition over a non-empty subset of $n$-qubit computational basis states is an important and challenging step in many quantum computation algorithms and applications. 
In this work, we address the problem of preparation of a uniform superposition state of the form  $\ket{\Psi} = \frac{1}{\sqrt{M}}\sum_{j = 0}^{M - 1} \ket{j}$, where $M$ denotes the number of distinct states in the superposition state and $2 \leq M \leq 2^n$.  We show that the superposition state $\ket{\Psi}$ can be efficiently prepared with a gate complexity and circuit depth of only $O(\log_2~M)$ for all $M$. 
This demonstrates an exponential reduction in gate complexity in comparison to other existing approaches in the literature for the general case of this problem. Another advantage of the proposed approach is that it requires only $n=\ceil{\log_2~M}$ qubits. Furthermore, neither ancilla qubits nor any quantum gates with multiple controls are needed in our approach for creating the uniform superposition state $\ket{\Psi}$.
It is also shown that a broad class of nonuniform superposition states that involve a mixture of uniform superposition states can also be efficiently created with the same circuit configuration that is used for creating the uniform superposition state $\ket{\Psi}$ described earlier, but with modified parameters.

\end{abstract}
	
		\section*{Author Contributions}
	Authors Alok Shukla and Prakash Vedula collaborated on developing the initial concept and an early version of the algorithm. Alok then formulated the main algorithm (Algorithm 1 in the paper), which achieved an exponential speedup over the initial version. He also wrote the first draft of the paper, including Algorithm 1, its explanation, and the proof, along with contributing the results on nonuniform superposition in Section 3. Alok also contributed to the final review of the paper. Prakash reviewed the paper and contributed to Section 2.5, as well as writing parts of the Introduction and Conclusion sections.

\section{Introduction}\label{sec:intro}
    Quantum state preparation is important for many applications in quantum computing and quantum information processing ~\cite{nielsen2002quantum,wittek2014quantum,
kieferova2019simulating,
low2017optimal,
berry2015simulating,
childs2017quantum,
wiebe2012quantum,
childs2020quantum,
SHUKLA2023127708,
Shukla_Vedula_2022b}. It involves creation of specific quantum states that are often in superposition or entanglement, so that the benefits of quantum computation and quantum information processing over classical counterparts can be realized. 
    Preparation of arbitrary quantum states with $n$ qubits generally involves an exponential number of CNOT gates $O(2^n)$. Besides gate complexity, circuit depth also has exponential scaling with the number of qubits. 
    
    The uniform superposition states that involve the full set of computational basis states (i.e., $\ket{\Psi} = \frac{1}{\sqrt{M}} \sum_{j=0}^{M-1} \ket{j}$, where  $M=2^n$) play important roles in several quantum algorithms and often serve as a starting point for implementing these algorithms. Some examples and applications include  Deutsch-Jozsa algorithm \cite{deutsch1992rapid}, Bernstein-Vazirani algorithm \cite{bernstein1993quantum} and its probabilistic generalization \cite{Shukla_Vedula_PBV_2023},  Grover's quantum search algorithm \cite{grover1997quantum, shukla2019trajectory}, quantum phase estimation,
    Simon's algorithm \cite{simon1997power}, Shor's algorithm \cite{shor1999polynomial} etc. We note that the preparation of the uniform superposition state $\ket{\Psi} = \frac{1}{\sqrt{M}} \sum_{j=0}^{M-1} \ket{j}$, with  $M=2^n$, is  straightforward as it can be prepared using $n = \log_2~M$  Hadamard gates.
    
    Uniform superposition over a particular subset $S$ of computational basis states as $\ket{\Psi} = \frac{1}{\sqrt{M}} \sum_{j \in S} \ket{j}$, is also of interest in many applications and can be useful in generalization of some of the algorithms mentioned above to the cases where $M\neq 2^n$. For example, the amplitude amplification algorithm, which is a generalization of Grover's quantum search algorithm, can also work when $M \neq 2^n$, and the authors suggest that quantum Fourier transform
    can be used to an equal superposition in such cases (Ref.~\cite{Brassard_2002}, Page 3, fourth paragraph therein). 
    Further, we note that in \cite{Shukla_Vedula_PBV_2023}, a generalized version of the Bernstein-Vazirani algorithm was provided for determining multiple secret keys through a probabilistic oracle. The number of secret keys to be determined, say $M$, may or may not be of the form $M = 2^r$, with $r \in \NN$, therefore the preparation of a uniform superposition state $\ket{\Psi} = \frac{1}{\sqrt{M}} \sum_{j =0}^{M} \ket{j}$ for a positive integer $M$, with $2 \leq  M \leq 2^n$, is needed to implement the probabilistic oracle in the general case. 
     Another relevant example is the Quantum Byzantine Agreement (QBA) protocol. We note that the QBA protocol is important in distributed computing as it addresses the problem of achieving consensus among a group of quantum nodes even when some nodes exhibit arbitrary or malicious behavior. The  Quantum Byzantine Agreement (QBA) protocol involves the preparation of the GHZ state and another uniform superposition state of the form $  \ket{\Psi} = \frac{1}{\sqrt{n^3}} \sum_{j=0}^{n^3 - 1} \ket{j} $  \cite{ben2005fast}. Another important application of creating the state $\ket{\Psi}$ is the generation of random numbers. We note that the generation of genuine randomness with classical means is considered impossible, whereas     
     by exploiting the inherent probabilistic nature of quantum computing, genuine randomness can be achieved. The generation of randomness is important in cryptography, simulations, and many other scientific applications. 
     In this paper, we consider the problem of state preparation of such a uniform superposition state $\ket{\Psi}$.
For the sake of convenience, we will assume that the subset $S$ contains the basis states \{$\ket{0}$, $\ket{1}$, \ldots $\ket{M - 1 }$ \}, where $2<M<2^n$ for a given $n \in \NN$. Hence, our main objective is to develop an efficient approach for the preparation of a uniform quantum superposition state of the form $\ket{\Psi} = \frac{1}{\sqrt{M}} \sum_{j =0 }^{M-1} \ket{j}$. While such a state can be efficiently prepared via a straightforward approach for cases where $M=2^r$ using $r = \log_2~M$ Hadamard gates, there are no efficient approaches known in the literature for the case when $M$ is arbitrary, especially when $M\neq 2^r$. 
The approaches presented in previous works require the gate complexity of $O(M)$ for the preparation of such states in the general case \cite{gleinig2021efficient, mozafari2021efficient, Qiskit}.

In this paper, we propose an efficient approach for quantum state preparation of uniform superposition state $\ket{\Psi}= \frac{1}{\sqrt{M}} \sum_{j =0 }^{M-1} \ket{j}$ that offers a significant (exponential) reduction in gate complexity and circuit depth without the use of ancillary qubits. We show that using only $n=\ceil{\log_2 ~M}$ qubits, the uniform superposition state $\ket{\Psi}$ can be prepared for arbitrary $M$ with a gate complexity and circuit depth of $O(\log_2 ~M)$.
Additionally, our proposed method in Algorithm \ref{alg_uniform_superposition} does not require quantum gates with multiple controls. Instead, only specific combinations of single qubit gates such as Pauli X gates, Hadamard gates, and rotation ($R_Y(\theta)$) gates, along with controlled gates (namely controlled Hadamard gates and controlled rotation gates) with a single control qubit are used. We observe that the controlled Hadamard gates and controlled rotation gates can be implemented using CNOT gates and a few single qubit gates. 
We demonstrate that (in the general case) our proposed approach achieves an exponential reduction in the number of CNOT gates needed compared to the Qiskit implementation (ref.~Table \ref{tab:quiskit}, \mfig{fig:mcases_all} and \mfig{fig:usp_cnot}).

Further, we show that this exponential reduction in complexity can also be extended to address the problem of quantum state preparation of a broad class of nonuniform superposition states that involve a mixture of uniform superpositions over multiple subsets of basis states. In other words, the quantum circuit configurations with $O(\log_2 ~M)$ gate complexity and circuit depth used for the generation of uniform superposition states $\ket{\Psi}$ for any given $M$ can be reused with
modified parameters or rotation angles (associated with rotation gates and controlled rotation gates) to generate a broad class of nonuniform superposition states.

The rest of this paper is organized as follows. 
A quantum algorithm for the preparation of uniform superposition state $\ket{\Psi} = \frac{1}{\sqrt{M}}\sum_{j = 0}^{M - 1} \ket{j}$, where $ M $ is a positive integer with $ 2 < M < 2^n $ and $ M \neq 2^r $ for any $ r \in \NN $, is given in \mref{sec:algo}. A detailed explanation of Algorithm \ref{alg_uniform_superposition} is provided in \ref{sec:explanation}.
    In \mref{sec_correctness_algo_one}, the correctness of Algorithm \ref{alg_uniform_superposition} is established.
Example quantum circuits are provided in \mref{sec:examples_usp} to illustrate how Algorithm \ref{alg_uniform_superposition} works. A detailed analysis of the complexity of Algorithm \ref{alg_uniform_superposition} is provided in \mref{sec:complexity_usp}. Quantum state preparation of a class of nonuniform superposition states using a variation of Algorithm \ref{alg_uniform_superposition} is described in \mref{sec:nusp} and some illustrative examples along with relevant quantum circuits are given in \mref{sec:nusp_examples}. Finally, the conclusion of the article is summarized in  \mref{sec:conclusion}.

\section{Uniform superposition}
\label{sec:usp}

\subsection{Algorithm}	\label{sec:algo}
One of the main objectives of this work is to consider the problem of preparation of the uniform quantum superposition state $ \frac{1}{\sqrt{M}} \sum_{j \in S}  \ket{j} $, where $ S $ is a subset of $n$-qubit basis states of the cardinality $ M $, with $2 \leq M \leq 2^n$.  
If $M = 2^r$, for $ 1 \leq r \leq n$, then one can use $r$ Hadamard gates to create the desired uniform superposition state. Therefore, in the rest of the paper, we assume that $ 2 < M < 2^n $ and $ M \neq 2^r $ for any $ r \in \NN $. Moreover, it will be convenient to consider $S$ to be the subset $ \{\ket{0}, \ket{1}, \ldots, \ket{M-1}  \}$. Therefore, to summarize, our goal is to prepare the uniform superposition state $\ket{\Psi} = \frac{1}{\sqrt{M}}\sum_{k = 0}^{M - 1} \ket{k}$, where $ 2 < M < 2^n $ and $ M \neq 2^r $ for any $ r \in \NN $. In other words, we want to create a quantum circuit using elementary quantum gates whose action can be represented as the unitary operator $\operatorname{U}$ such that $\operatorname{U} \ket{0}^{\otimes n} = \ket{\Psi} $. 

In Algorithm \ref{alg_uniform_superposition}, a quantum circuit to prepare the desired uniform superposition state $\ket{\Psi} = \frac{1}{\sqrt{M}}\sum_{k = 0}^{M - 1} \ket{k}$ is provided. 
We note that Algorithm \ref{alg_uniform_superposition} can create the uniform superposition state $\ket{\Psi} = \frac{1}{\sqrt{M}}\sum_{j = 0}^{M - 1} \ket{j}$ extremely efficiently by using only $\ceil{\log_2~ M}$ qubits.
It employs Hadamard ($H$), controlled Hadamard, rotation ($R_{Y}(\theta)$)  and controlled rotation gates in the general case. 
We note that the operator $ R_{Y}(\theta) $ represents the rotation (through an angle $\theta  $) about the $ Y $-axis of the Bloch sphere representation. The unitary matrix corresponding to this operator is
	\[
	R_{Y}(\theta) = 
	\begin{bmatrix}
		\cos {\frac{\theta}{2}} & - \sin {\frac{\theta}{2}} \\ 
		&	 \\
		\sin {\frac{\theta}{2}} & \cos {\frac{\theta}{2}}
	\end{bmatrix}.
	\]  

\begin{algorithm}[H] \label{alg_uniform_superposition}
	\DontPrintSemicolon
	\KwInput{A positive integer $ M $ with $ 2 < M < 2^n $ and $ M \neq 2^r $ for any $ r \in \NN $. }
	\KwOutput{A quantum state $ 	U_{M} \ket{0}^{\otimes n}  =    \frac{1}{\sqrt{M}}  \,  \sum_{j=0}^{M-1} \, \ket{j}, $ that is in a uniform superposition of $ M $ distinct states.} 
	\Fn{Uniform (M)}{
        \tcc{ $  l_0, l_1, \ldots\,,l_k  $ is an ordered sequence of numbers representing the locations of $ 1 $ in the reverse binary representation of $ M $.}
		Compute $  l_0, l_1, \ldots\,,l_k  $, where $ M = \sum_{j=0}^{k} \, 2^{l_j} $ with $ 0 \leq l_0 < l_1 < \ldots < l_{k-1} < l_k \leq n-1 $. \\
		Initialize $ \ket{\Psi} =   \ket{q_{n-1}  } \otimes \ket{q_{n -2 }} \otimes \cdots \cdots \otimes \ket{q_{1}} \otimes \ket{q_{0}}  = \ket{0}^{\otimes n} $.\\
		Apply $ X $ gate on $ \ket{q_{i}} $ for $ i = l_1$, $ l_2$, $ \ldots $, $ l_k $. \tcp*{Apply $ X $ gates on  qubits at positions $ l_1 $, $ l_2$, $ \ldots $, $ l_k $. }
		Set $ M_0 = 2^{l_0} $. \\
        \tcc{If $ M $ is an even number, then apply Hadamard gates on the rightmost $l_0$ qubits.} 
		\If{ $ l_0 > 0  $} 
		{Apply $ H $ gate on  $ \ket{q_{i}} $ for $ i = 0$, $ 1$, $ \ldots $, $ l_0 -1 $.}
		Apply the rotation gate $ R_Y(\theta_0) $ on $ \ket{q_{l_1}}$, where  $ \theta_0 = - 2 \arccos \left(\sqrt {\frac{M_0}{M}}\right)  $. \\
		Apply a controlled Hadamard ($ H $) gate on $ \ket{q_{i}} $ for $ i = l_0$, $ l_0 + 1$, $ \ldots $, $ l_1 -1 $ conditioned on $  q_{l_1}$ being equal to $ 0 $. \\
		\For{$ m=1 $ to $ k -1$}{
		Apply a controlled $ R_Y(\theta_m) $ gate, with  $ \theta_m = - 2 \arccos \left(\sqrt {\frac{2^{l_m}}{M-M_{m-1}}}\right)  $, on $ \ket{{q}_{l_{m+1}}} $  conditioned on $  q_{l_m}$ being $ 0 $. \\
		Apply a controlled Hadamard ($ H $) gate on $ \ket{q_{i}} $ for $ i = l_m$, $ l_m + 1$, $ \ldots $, $ l_{m+1} -1 $ conditioned on $ {{q}_{l_{m+1}}} $  being equal to $ 0 $.\\
		Set $ M_m = M_{m-1} + 2^{l_m} $.
	      }
	\Return{$ \ket{\Psi}$}
}
	\caption{A quantum algorithm for the preparation of uniform superposition state $\ket{\Psi} = \frac{1}{\sqrt{M}}\sum_{j = 0}^{M - 1} \ket{j}$.}
\end{algorithm}
\vskip 0.5cm
\subsection{Explanation of the algorithm} \label{sec:explanation}
Let $ M $ be a positive integer with $ 2 < M < 2^n $ and $ M \neq 2^r $ for any $ r \in \NN $, i.e., $ M $ is not an integer power of $ 2 $.
Let $ M = \sum_{j=0}^{k} \, 2^{l_j} $ with $ 0 \leq l_0 < l_1 < \ldots < l_{k-1} < l_k \leq n-1 $. It is clear that 
$  l_0$, $ l_1 $, $ \ldots $, $ l_k $,  is an ordered sequence of numbers that contain the locations of $ 1 $ in the reverse
binary representation of $ M $. Algorithm \ref{alg_uniform_superposition} begins by computation of  $  l_0, l_1, \ldots\,,l_k  $, followed by 
initialization of the quantum state $ \ket{\Psi} =   \ket{q_{n-1}  } \otimes \ket{q_{n -2 }} \otimes \cdots \cdots \otimes \ket{q_{1}} \otimes \ket{q_{0}}  = \ket{0}^{\otimes n} $.

Let for any integer $r$, with $ 0 \leq r \leq   k $, $M_r$ be defined as     
\begin{align} \label{eq:M_r}
  M_r =   \sum_{j=0}^{r} \, 2^{l_j}.
\end{align}
We note that $M_0$ is defined in line $3$ in Algorithm \ref{alg_uniform_superposition}. Further, line $13$ of Algorithm \ref{alg_uniform_superposition} iteratively defines $M_r$ for $r=1$
to $r = k-1$. 

The key steps in Algorithm \ref{alg_uniform_superposition} are lines $11$ and $12$, involving the applications of a controlled rotation and controlled Hadamard gates.
In line $11$, $ R_Y(\theta_m) $ gate, with  $ \theta_m = - 2 \arccos \left(\sqrt {\frac{2^{l_m}}{M-M_{m-1}}}\right)  $, acts on  on $ \ket{{q}_{l_{m+1}}} $  conditioned on $  q_{l_m}$ being $ 0 $. 
Further, the action of the rotation gate $R_Y(\theta_r)$, with  $ \theta_r = - 2 \arccos \left(\sqrt {\frac{2^{l_r}}{M- M_{r-1} }}\right)  $, on $ \ket{1}$  is the following,
\begin{align}
    R_Y(\theta_r)  \ket{1}   = a_r \ket{0} + b_r \ket{1},
\end{align}
where
\begin{align} \label{eq_a_r_b_r}
		b_r = \sqrt {\frac{2^{l_r}    }{M -    M_{r-1} }},  \quad \text{and} \quad     
	    a_r = \sqrt {\frac{M -    M_{r} } {M -  M_{r-1} }},	
\end{align}
with $  0 < r \leq k-1 $ and 
\begin{align} \label{eq_a_zero_b_zero}
b_0 = \sqrt {\frac{2^{l_0}}{M}},  \quad \text{and} \quad a_0 = \sqrt {\frac{M- 2^{l_0}}{M}}.	
\end{align}
Clearly, the coefficients $a_r$ and $b_r$ satisfy the normalization condition $\abs{a_r}^2  + \abs{b_r}^2 =1$.
Next, we describe the steps of Algorithm \ref{alg_uniform_superposition} in detail.\\
\noindent \textbf{Case 1: $M$ is odd.}
First, we consider the case when $ M $ is odd. It means $l_0 = 0$. 
On the application of the $ X $ gate (ref.~line $ 4 $, Algorithm \ref{alg_uniform_superposition}) on $ \ket{q_{i}} $ for $ i = l_1$, $ l_2$, $ \ldots $, $ l_k $, the following quantum state is obtained,
\begin{align}
       | \, 0  \, \cdots \,  \smallunderbrace{1}_{ l_k }  \, \cdots  \, 0  \, \cdots \, \smallunderbrace{1}_{ l_{k-1} }  \, \cdots \, 0  \, \cdots \,  \smallunderbrace{1}_{ l_2 }  \, \cdots \, 0  \, \cdots \,   \smallunderbrace{1}_{ l_1 }  \, \cdots \, 0  \, \cdots \,  \smallunderbrace{0}_{ l_0 =0 }  \rangle.
\end{align}
Here, the notation $  \smallunderbrace{1}_{ l_r }  \, \cdots  \, 0  \, \cdots \, \smallunderbrace{1}_{ l_{r-1} }  $ indicates that all the qubits between the positions $l_{r-1} $ and $l_r$ are in the quantum state $\ket{0}$, and the notation $| \, 0  \, \cdots \,  \smallunderbrace{1}_{ l_k }$ represents the fact that all the qubits on the left of $l_k$ are in the quantum state $\ket{0}$. This convention will be followed in the rest of the paper. 

Next, the action of the rotation $ R_Y(\theta_0) $ gate (ref.~line $ 8 $, Algorithm \ref{alg_uniform_superposition}) on $ \ket{q_{l_1}}$, with $  \theta_0 = - 2 \arccos \left(\sqrt {\frac{M_0}{M}}\right)  $, where $ M_0 = 2^{l_0} $, gives the quantum state
\begin{align}
	  & b_0 | \, 0  \, \cdots \,  \smallunderbrace{1}_{ l_k }  \, \cdots  \, 0  \, \cdots \, \smallunderbrace{1}_{ l_{k-1} }  \, \cdots \, 0  \, \cdots \,  \smallunderbrace{1}_{ l_2 }  \, \cdots \, 0  \, \cdots \,   \smallunderbrace{1}_{ l_1 }  \, \cdots \, 0  \, \cdots \,  \smallunderbrace{0}_{ l_0 =0 }  \rangle \nonumber \\
		& + 
		a_0   | \, 0  \, \cdots \,  \smallunderbrace{1}_{ l_k }  \, \cdots  \, 0  \, \cdots \, \smallunderbrace{1}_{ l_{k-1} }  \, \cdots \, 0  \, \cdots \,  \smallunderbrace{1}_{ l_2 }  \, \cdots \, 0  \, \cdots \,   \smallunderbrace{0}_{ l_1 }  \, \cdots \, 0  \, \cdots \,  \smallunderbrace{0}_{ l_0 =0 }  \rangle,
\end{align}
where 
$a_0$ and $b_0$ are as defined earlier in  \meqref{eq_a_zero_b_zero}.
Subsequently, the application of the controlled Hadamard gate (ref.~line $ 9 $, Algorithm \ref{alg_uniform_superposition}) on $ \ket{q_{i}} $ for $ i = l_0$, $ l_0 + 1$, $ \ldots $, $ l_1 -1 $ conditioned on $  q_{l_1}$ being equal to $ 0 $ gives the quantum state,
 \begin{align}
	\ket{\psi_0}  &= b_0 | \, 0  \, \cdots \,  \smallunderbrace{1}_{ l_k }  \, \cdots  \, 0  \, \cdots \, \smallunderbrace{1}_{ l_{k-1} }  \, \cdots \, 0  \, \cdots \,  \smallunderbrace{1}_{ l_2 }  \, \cdots \, 0  \, \cdots \,   \smallunderbrace{1}_{ l_1 }  \, \cdots \, 0  \, \cdots \,  \smallunderbrace{0}_{ l_0 =0 }  \nonumber \rangle \\
	& + 
	a_0   | \, 0  \, \cdots \,  \smallunderbrace{1}_{ l_k }  \, \cdots  \, 0  \, \cdots \, \smallunderbrace{1}_{ l_{k-1} }  \, \cdots \, 0  \, \cdots \,  \smallunderbrace{1}_{ l_2 }  \, \cdots \, 0  \, \cdots \,   \smallunderbrace{0}_{ l_1 }  \, \cdots \, +  \, \cdots \,  \smallunderbrace{+}_{ l_0 =0 }  \rangle,
\end{align}
where $ \ket{+} = \frac{1}{\sqrt{2}} \left( \ket{0} +  \ket{1} \right) $.
Next we consider the ``For Loop'' in lines $ 10 $-$ 13 $ in Algorithm \ref{alg_uniform_superposition}. We note that in the first iteration, i.e., for $ m=1 $, the  
application of a controlled rotation on $ \ket{{q}_{l_{2}}} $  conditioned on $  q_{l_1}$ being $ 0 $ (line $ 11 $, Algorithm \ref{alg_uniform_superposition}) results in the quantum state
\begin{align}
	  & b_0 | \, 0  \, \cdots \,  \smallunderbrace{1}_{ l_k }  \, \cdots  \, 0  \, \cdots \, \smallunderbrace{1}_{ l_{k-1} }  \, \cdots \, 0  \, \cdots \,  \smallunderbrace{1}_{ l_2 }  \, \cdots \, 0  \, \cdots \,   \smallunderbrace{1}_{ l_1 }  \, \cdots \, 0  \, \cdots \,  \smallunderbrace{0}_{ l_0 =0 }  \rangle  \nonumber\\
	& + 
	a_0 b_1  | \, 0  \, \cdots \,  \smallunderbrace{1}_{ l_k }  \, \cdots  \, 0  \, \cdots \, \smallunderbrace{1}_{ l_{k-1} }  \, \cdots \, 0  \, \cdots \,  \smallunderbrace{1}_{ l_2 }  \, \cdots \, 0  \, \cdots \,   \smallunderbrace{0}_{ l_1 }  \, \cdots \, +  \, \cdots \,  \smallunderbrace{+}_{ l_0 =0 }  \rangle  \nonumber\\
   & +	a_0 a_1  | \, 0  \, \cdots \,  \smallunderbrace{1}_{ l_k }  \, \cdots  \, 0  \, \cdots \, \smallunderbrace{1}_{ l_{k-1} }  \, \cdots \, 0  \, \cdots \,  \smallunderbrace{0}_{ l_2 }  \, \cdots \, 0  \, \cdots \,   \smallunderbrace{0}_{ l_1 }  \, \cdots \, +  \, \cdots \,  \smallunderbrace{+}_{ l_0 =0 }  \rangle.
\end{align}
Then application of a controlled Hadamard ($ H $) gate on $ \ket{q_{i}} $ for $ i = l_1$, $ l_1 + 1$, $ \ldots $, $ l_2 -1 $ conditioned on $ {{q}_{l_{2}}} $  being equal to $ 0 $ (ref.~line $ 12 $, Algorithm \ref{alg_uniform_superposition}), results in the quantum state 
\begin{align}
\ket{\psi_1} = 	& b_0 | \, 0  \, \cdots \,  \smallunderbrace{1}_{ l_k }  \, \cdots  \, 0  \, \cdots \, \smallunderbrace{1}_{ l_{k-1} }  \, \cdots \, 0  \, \cdots \,  \smallunderbrace{1}_{ l_2 }  \, \cdots \, 0  \, \cdots \,   \smallunderbrace{1}_{ l_1 }  \, \cdots \, 0  \, \cdots \,  \smallunderbrace{0}_{ l_0 =0 }  \rangle  \nonumber\\
	& + 
	a_0 b_1  | \, 0  \, \cdots \,  \smallunderbrace{1}_{ l_k }  \, \cdots  \, 0  \, \cdots \, \smallunderbrace{1}_{ l_{k-1} }  \, \cdots \, 0  \, \cdots \,  \smallunderbrace{1}_{ l_2 }  \, \cdots \, 0  \, \cdots \,   \smallunderbrace{0}_{ l_1 }  \, \cdots \, +  \, \cdots \,  \smallunderbrace{+}_{ l_0 =0 }  \rangle  \nonumber\\
	& +	a_0 a_1  | \, 0  \, \cdots \,  \smallunderbrace{1}_{ l_k }  \, \cdots  \, 0  \, \cdots \, \smallunderbrace{1}_{ l_{k-1} }  \, \cdots \, 0  \, \cdots \,  \smallunderbrace{0}_{ l_2 }  \, \cdots \, +  \, \cdots \,   \smallunderbrace{+}_{ l_1 }  \, \cdots \, +  \, \cdots \,  \smallunderbrace{+}_{ l_0 =0 }  \rangle.
\end{align}
Here, 
\begin{align*}
	b_1 = \sqrt {\frac{2^{l_1}}{M - 2^{l_0}}},  \quad \text{and} \quad     a_1 = \sqrt {\frac{M- 2^{l_0} - 2^{l_1} }{M-2^{l_0} }},	
\end{align*}
are as defined in \meqref{eq_a_r_b_r}.
It follows from an easy induction argument (see \mref{sec_correctness_algo_one}) that at the end of the iteration $ m = k-1 $, the following quantum state is obtained,
\begin{align} \label{eq_psi_k_final}
	\ket{\psi_{k-1}} = 	& b_0 | \, 0  \, \cdots \,  \smallunderbrace{1}_{ l_k }  \, \cdots  \, 0  \, \cdots \, \smallunderbrace{1}_{ l_{k-1} }  \, \cdots \, 0  \, \cdots \,  \smallunderbrace{1}_{ l_2 }  \, \cdots \, 0  \, \cdots \,   \smallunderbrace{1}_{ l_1 }  \, \cdots \, 0  \, \cdots \,  \smallunderbrace{0}_{ l_0 =0 }  \rangle  \nonumber\\
	& + 
	a_0 b_1  | \, 0  \, \cdots \,  \smallunderbrace{1}_{ l_k }  \, \cdots  \, 0  \, \cdots \, \smallunderbrace{1}_{ l_{k-1} }  \, \cdots \, 0  \, \cdots \, \smallunderbrace{1}_{ l_{3} }  \, \cdots \, 0  \, \cdots \,  \smallunderbrace{1}_{ l_2 }  \, \cdots \, 0  \, \cdots \,   \smallunderbrace{0}_{ l_1 }  \, \cdots \, +  \, \cdots \,  \smallunderbrace{+}_{ l_0 =0 }  \rangle \nonumber \\
	& +	a_0 a_1 b_2 | \, 0  \, \cdots \,  \smallunderbrace{1}_{ l_k }  \, \cdots  \, 0  \, \cdots \, \smallunderbrace{1}_{ l_{k-1} }  \, \cdots  \, 0  \, \cdots \, \smallunderbrace{1}_{ l_{3} }  \, \cdots \, 0  \, \cdots \,  \smallunderbrace{0}_{ l_2 }  \, \cdots \, +  \, \cdots \,   \smallunderbrace{+}_{ l_1 }  \, \cdots \, +  \, \cdots \,  \smallunderbrace{+}_{ l_0 =0 }  \rangle \nonumber \\
	&+	a_0 a_1 a_2 b_3 | \, 0  \, \cdots \,  \smallunderbrace{1}_{ l_k }  \, \cdots  \, 0  \, \cdots \, \smallunderbrace{1}_{ l_{k-1} }  \, \cdots  \, 0  \, \cdots \, \smallunderbrace{0}_{ l_{3} }  \, \cdots \, +  \, \cdots \,  \smallunderbrace{+}_{ l_2 }  \, \cdots \, +  \, \cdots \,   \smallunderbrace{+}_{ l_1 }  \, \cdots \, +  \, \cdots \,  \smallunderbrace{+}_{ l_0 =0 }  \rangle \nonumber \\
	& \cdots  \nonumber\\
	&+	a_0 a_1 \ldots a_{k-2} b_{k-1} | \, 0  \, \cdots \,  \smallunderbrace{1}_{ l_k }  \, \cdots  \, 0  \, \cdots \, \smallunderbrace{0}_{ l_{k-1} }  \, \cdots \, +  \, \cdots \, \smallunderbrace{+}_{ l_{3} }  \, \cdots \, +  \, \cdots \,  \smallunderbrace{+}_{ l_2 }  \, \cdots \, +  \, \cdots \,   \smallunderbrace{+}_{ l_1 }  \, \cdots \, +  \, \cdots \,  \smallunderbrace{+}_{ l_0 =0 }  \rangle \nonumber \\
	&+	a_0 a_1 \ldots a_{k-1} | \, 0  \, \cdots \,  \smallunderbrace{0}_{ l_k }  \, \cdots  \, + \, \cdots \, \smallunderbrace{+}_{ l_{k-1} }  \, \cdots \, +  \, \cdots \,  \smallunderbrace{+}_{ l_{3} }  \, \cdots \, +  \, \cdots \, \smallunderbrace{+}_{ l_2 }  \, \cdots \, +  \, \cdots \,   \smallunderbrace{+}_{ l_1 }  \, \cdots \, +  \, \cdots \,  \smallunderbrace{+}_{ l_0 =0 }  \rangle,
\end{align}
where $a_r$ and $b_r$ are defined in \meqref{eq_a_r_b_r} and \meqref{eq_a_zero_b_zero}.
We observe that the following has $ l_r $ qubits in the state $ \ket{+} $, therefore, 
\begin{align}\label{eq_plus_states}
		a_0 a_1 \ldots a_{r-1} b_{r} | \, 0  \, \cdots \,  \smallunderbrace{1}_{ l_k }  \, \cdots  \, 0  \, \cdots \, \smallunderbrace{0}_{ l_{r} }  \, \cdots \, +  \, \cdots \, \smallunderbrace{+}_{ l_{3} }  \, \cdots \, +  \, \cdots \,  \smallunderbrace{+}_{ l_2 }  \, \cdots \, +  \, \cdots \,   \smallunderbrace{+}_{ l_1 }  \, \cdots \, +  \, \cdots \,  \smallunderbrace{+}_{ l_0 =0 }  \rangle 
\end{align}
contains a superposition of $ 2^{l_r} $ quantum states with equal amplitude $$ \frac{a_0 a_1 \ldots a_{r-1} b_{r}}{\sqrt{2^{l_r}}} $$ for $  0 < r \leq k-1 $. 
It is easy to see that
\begin{align}\label{eq_all_equal}
	\frac{ b_0 }{\sqrt{2^{l_0}}} = \frac{a_0 b_1 }{\sqrt{2^{l_1}}} =   \frac{a_0 a_1 b_2}{\sqrt{2^{l_2}}} = \frac{a_0 a_1 a_2 b_3}{\sqrt{2^{l_3}}}  = \cdots \cdots =  \frac{a_0 a_1 \ldots a_{k-2} b_{k-1}}{\sqrt{2^{l_{k-1}}}}  =  \frac{a_0 a_1 \ldots a_{k-2} a_{k-1} }{\sqrt{2^{l_{k}}}} = \frac{1}{\sqrt{M}}.
\end{align}
From \meqref{eq_all_equal} and \meqref{eq_psi_k_final}, it follows that the output of Algorithm \ref{alg_uniform_superposition} is a uniform superposition of $ M $ distinct states, as desired. 

\subparagraph{Case 2: $M$ is even.}
If $M$ is an even number, then it is clear that $ l_0 \neq 0$. On the application of the $ X $ gate (line $ 4 $, Algorithm \ref{alg_uniform_superposition}) on $ \ket{q_{i}} $ for $ i = l_1$, $ l_2$, $ \ldots $, $ l_k $, the following state is obtained.
\begin{align}
    \mathclap{   | \, 0  \, \cdots \,  \smallunderbrace{1}_{ l_k }  \, \cdots  \, 0  \, \cdots \, \smallunderbrace{1}_{ l_{k-1} }  \, \cdots \, 0  \, \cdots \,  \smallunderbrace{1}_{ l_2 }  \, \cdots \, 0  \, \cdots \,   \smallunderbrace{1}_{ l_1 }  \, \cdots \, 0  \, \cdots \,  \smallunderbrace{0}_{ l_0  } \, \cdots \, 0  \, \cdots \,  \rangle.}
\end{align}
Since $ l_0 > 0$ in this case, the Hadamard gates are applied on $ \ket{q_{i}} $ for $ i = 0$, $ 1$, $ \ldots $, $ l_0-1 $, and the following state is obtained (ref.~lines $ 5$ and $6$, Algorithm \ref{alg_uniform_superposition}).
\begin{align}
    \mathclap{   | \, 0  \, \cdots \,  \smallunderbrace{1}_{ l_k }  \, \cdots  \, 0  \, \cdots \, \smallunderbrace{1}_{ l_{k-1} }  \, \cdots \, 0  \, \cdots \,  \smallunderbrace{1}_{ l_2 }  \, \cdots \, 0  \, \cdots \,   \smallunderbrace{1}_{ l_1 }  \, \cdots \, 0  \, \cdots \,  \smallunderbrace{0}_{ l_0  } \, \cdots \, +  \, \cdots \,  \rangle.}
\end{align}
When $ M $ is an even number, the remaining steps of Algorithm \ref{alg_uniform_superposition}  are similar to the previous case (i.e., when $M$ is odd) and can be easily verified. 

\subsection{Proof of the correctness of Algorithm \ref{alg_uniform_superposition}} \label{sec_correctness_algo_one}
Steps of Algorithm \ref{alg_uniform_superposition} were already described in \mref{sec:explanation}. It only remains to show that the ``For Loop'' in lines $ 10 $-$ 13 $ in Algorithm \ref{alg_uniform_superposition} works correctly. In the following, mathematical induction will be used to prove this.

It is easy to see that at the end of the iteration $ m = 1 $, the quantum state is 
\begin{align}
    \ket{\psi_1} = \frac{b_0 }{\sqrt{2^{l_0}}}  \sum_{j=0}^{2^{l_0} - 1} \ket{j + M - 2^{l_0}} &+  \frac{a_0 b_1 }{\sqrt{2^{l_1}}}  \sum_{j=0}^{2^{l_1} - 1} \ket{j + M - 2^{l_0} - 2^{l_1}} + \frac{a_0 a_1  }{\sqrt{2^{l_{2}}}}\sum_{j=0}^{2^{l_{2}} - 1} \ket{  j + M - \sum_{s=0}^{2} \, 2^{l_s}},
\end{align}
where $a_1$, $b_1$, $a_0$ and $b_0$ are as defined in \meqref{eq_a_r_b_r} and \meqref{eq_a_zero_b_zero}, respectively.

Let $ 2 < r \leq k-1$. Assume that at the end of the iteration, $ m = r -1 $ (or at the beginning of the 
iteration $m=r$), the quantum state obtained is
\begin{align}\label{eq_psi_algo_one_general}
\ket{\psi_{r-1}} = 
    \frac{b_0 }{\sqrt{2^{l_0}}}  \sum_{j=0}^{2^{l_0} - 1} \ket{j + M - 2^{l_0}} &+  \frac{a_0 b_1 }{\sqrt{2^{l_1}}}  \sum_{j=0}^{2^{l_1} - 1} \ket{j + M - 2^{l_0} - 2^{l_1}}
    + \frac{a_0 a_1 b_2 }{\sqrt{2^{l_2}}}   \sum_{j=0}^{2^{l_2} - 1} \ket{j + M - 2^{l_0} - 2^{l_1} - 2^{l_2} } \nonumber\\
    \cdots \cdots
    &+ \frac{a_0 a_1 \ldots a_{r-2} b_{r-1}}{\sqrt{2^{l_{r-1}}}} \sum_{j=0}^{2^{l_{r-1}} - 1} \ket{j + M - \sum_{s=0}^{r-1} \, 2^{l_s} }
    + \frac{a_0 a_1 \ldots a_{r-2} a_{r-1} }{\sqrt{2^{l_{r}}}}\sum_{j=0}^{2^{l_{r}} - 1} \ket{  j + M - \sum_{s=0}^{r} \, 2^{l_s}}
\end{align}
where $a_r$ and $b_r$ are defined in \meqref{eq_a_r_b_r} and \meqref{eq_a_zero_b_zero}.
It follows from \meqref{eq_a_r_b_r} and \meqref{eq_a_zero_b_zero} that $\ket{\psi_{r-1}}$ can alternatively be written as
\begin{align}
    \ket{\psi_{r-1}} = \frac{1}{\sqrt{M}} \left( \sum_{j= 0}^{M_{r-1}-1} \, \ket{j +M - M_{r-1}} \right) + 
    \sqrt {\frac{M -M_{r-1} } {M 2^{l_r}}}  \left(\sum_{j=0}^{2^{l_{r}} - 1} \ket{  j + M - M_r}\right),
\end{align}
where $M_r$ is defined in \eqref{eq:M_r}.

It follows from Lemma \ref{lemma_induction_algo_one} that the actions of lines $11$ and $12$ of Algorithm \ref{alg_uniform_superposition} produce the state $\ket{\psi_r}$ as given in 
\meqref{eq_psi_r}. This completes the proof of the induction step. On taking $r = k-1 $ in \meqref{eq_psi_r} the state
$ \ket{\psi_{k-1}} = \frac{1}{\sqrt{M}}  \,  \sum_{j=0}^{M-1} \, \ket{j}$ is obtained, that is in a uniform superposition of $ M $ distinct quantum states, proving the correctness of Algorithm \ref{alg_uniform_superposition}.  

\begin{lem}\label{lemma_induction_algo_one}
Let $M$ be as defined in Algorithm \ref{alg_uniform_superposition}, i.e., $ M = \sum_{j=0}^{k} \, 2^{l_j} $ with $ 0 \leq l_0 < l_1 < \ldots < l_{k-1} < l_k \leq n-1 $.
    Let \begin{align}
    \ket{\psi_{r-1}} = \frac{1}{\sqrt{M}} \left( \sum_{j= 0}^{M_{r-1}-1} \, \ket{j +M - M_{r-1}} \right) + 
    \sqrt {\frac{M -M_{r-1} } {M 2^{l_r}}}  \left(\sum_{j=0}^{2^{l_{r}} - 1} \ket{  j + M - M_r}\right),
\end{align}
with $ 2 \leq r \leq k-1$ and where $M_r$ is defined in \meqref{eq:M_r}. Then the actions of lines $11$ and $12$ of Algorithm \ref{alg_uniform_superposition} (i.e., the action of a controlled $ R_Y(\theta_r) $ gate, with  $ \theta_r = - 2 \arccos \left(\sqrt {\frac{2^{l_r}}{M- M_{r-1} }}\right)  $, on $ \ket{{q}_{l_{r+1}}} $  conditioned on $  q_{l_r}$ being $ 0 $, followed by the action of the controlled Hadamard ($ H $) gate on $ \ket{q_{i}} $ for $ i = l_r$, $ l_r + 1$, $ \ldots $, $ l_{r+1} -1 $ conditioned on $ {{q}_{l_{r+1}}} $  being equal to $ 0 $), results in the following quantum state,
\begin{align}\label{eq_psi_r}
    \ket{\psi_{r}} = \frac{1}{\sqrt{M}} \left( \sum_{j= 0}^{M_{r}-1} \, \ket{j +M - M_{r}} \right) + 
    \sqrt {\frac{M -M_r } {M 2^{l_{r+1}}}}  \left(\sum_{j=0}^{2^{l_{r+1}} - 1} \ket{  j + M - M_{r+1}}\right).
\end{align}
\end{lem}
\begin{proof}
The  application of a controlled $ R_Y(\theta_r) $ gate, with  $ \theta_r = - 2 \arccos \left(\sqrt {\frac{2^{l_r}}{M- M_{r-1} }}\right)  $, on $ \ket{{q}_{l_{r+1}}} $  conditioned on $  q_{l_r}$ being $ 0 $ results in the following quantum state,
\begin{align}
     \frac{1}{\sqrt{M}} \left( \sum_{j= 0}^{M_{r-1}-1} \, \ket{j +M - M_{r-1}} \right) + 
    \sqrt {\frac{M - M_{r-1} } {M 2^{l_r}}}  \left( b_r \sum_{j=0}^{2^{l_{r}} - 1} \ket{  j + M - M_r}
    +  a_r \sum_{j=0}^{2^{l_{r}} - 1} \ket{  j + M - M_{r+1}}\right),
\end{align}
where $a_r$ and $b_r$ are as defined in \meqref{eq_a_r_b_r} and \meqref{eq_a_zero_b_zero}, respectively.
Then application of a controlled Hadamard ($ H $) gate on $ \ket{q_{i}} $ for $ i = l_r$, $ l_r + 1$, $ \ldots $, $ l_{r+1} -1 $ conditioned on $ {{q}_{l_{r+1}}} $  being equal to $ 0 $, yields the following quantum state,
\begin{align}
    \ket{\psi_{r}} & = \frac{1}{\sqrt{M}} \left( \sum_{j= 0}^{M_{r-1}-1} \, \ket{j +M - M_{r-1}} \right) + 
    \sqrt {\frac{M - M_{r-1}} {M 2^{l_r}}}  \left( b_r \sum_{j=0}^{2^{l_{r}} - 1} \ket{  j + M - M_r}
    +  \frac{a_r}{\sqrt{2^{l_{r+1} - l_r }}} \sum_{j=0}^{2^{l_{r+1}} - 1} \ket{  j + M - M_{r+1}}\right) \nonumber \\
    & = \frac{1}{\sqrt{M}} \left( \sum_{j= 0}^{M_{r}-1} \, \ket{j +M - M_{r}} \right) + 
    \sqrt {\frac{M - M_r } {M 2^{l_{r+1}}}}  \left(\sum_{j=0}^{2^{l_{r+1}} - 1} \ket{  j + M - M_{r+1}}\right).
\end{align}
This completes the proof.
\end{proof}

\subsection{Example circuits}
\label{sec:examples_usp}
\begin{example} \label{ex:usp_examples_ex_one}
We illustrate how Algorithm \ref{alg_uniform_superposition} works by considering the case of $ M =13$. Here $ M = 2^0 + 2^2 + 2^3$. Therefore, $l_0 = 0$, $l_1 = 2$ and $l_2 = 3$. We will use only $n = \ceil{\log_2~ M} = 4 $ qubits to create the uniform superposition state  $ \frac{1}{\sqrt{13}} \sum_{j=0}^{12} \, \ket{j}$.
In this case, the quantum circuit produced by Algorithm \ref{alg_uniform_superposition} is shown in \mfig{fig:uniform_ex_one}. In the following, we describe the steps of Algorithm \ref{alg_uniform_superposition} in detail for this case.

 \begin{figure}[H]
  \begin{center}
       \includegraphics[width=0.47\textwidth]{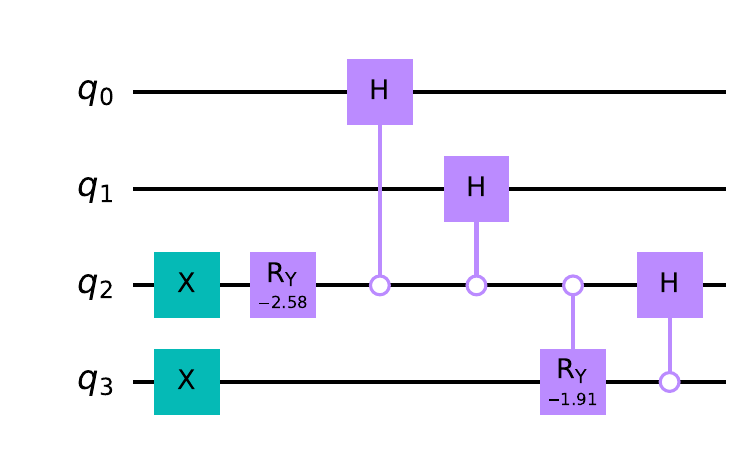}
        \includegraphics[width=0.52\textwidth]{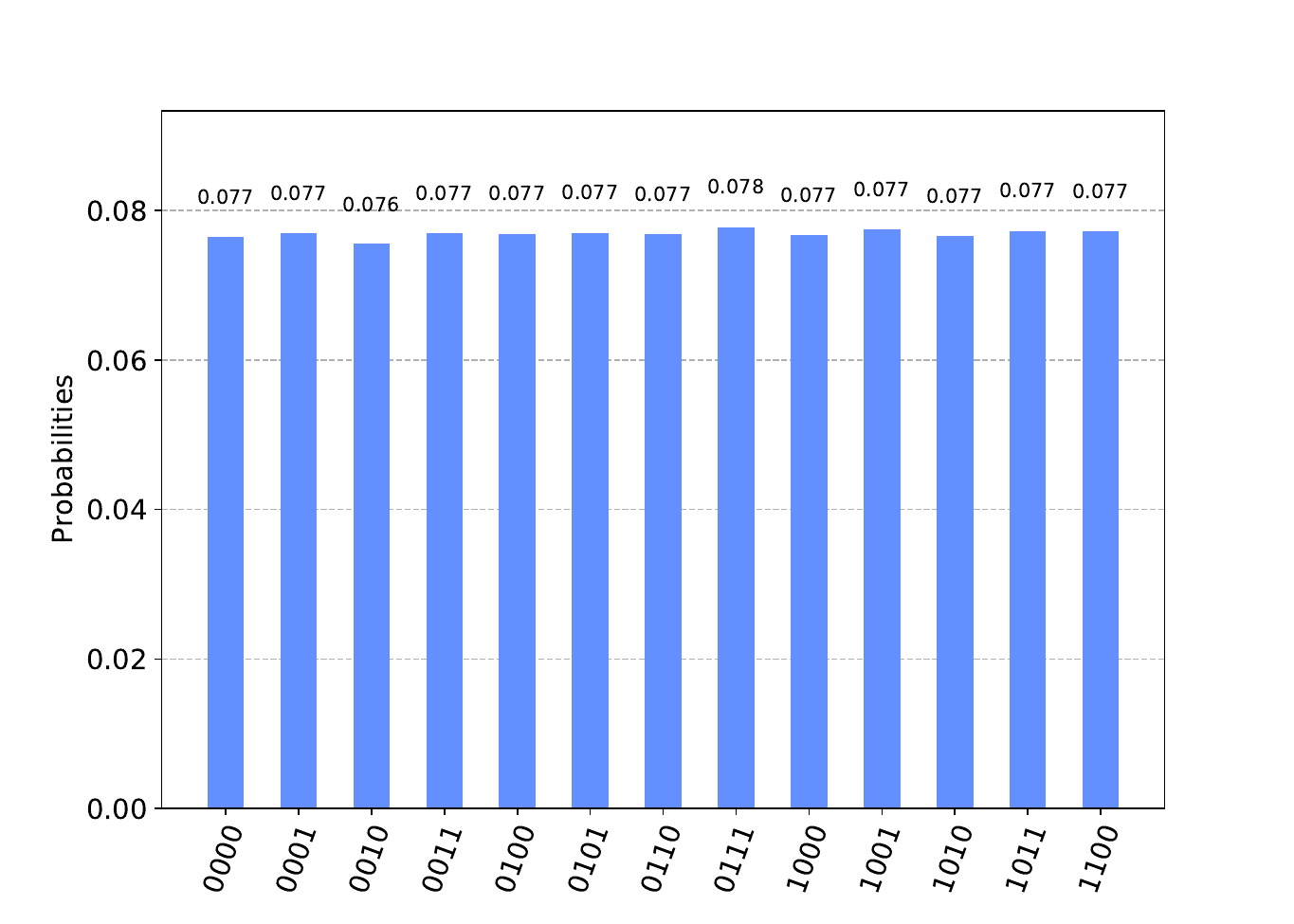}
  \end{center}
	 	\caption{A quantum circuit for creating the uniform quantum superposition state given in \meqref{eq_usp_state_vector_ex_one} is shown on the left. This corresponds to the case $ M=13 $ in Example \ref{ex:usp_examples_ex_one}. A histogram representing the sampling probabilities of obtaining various computational basis states is shown on the right.}
	 	\label{fig:uniform_ex_one}
	 \end{figure}

To begin with, each qubit is initialized to $\ket{0}$ (ref.~line $ 3 $, Algorithm \ref{alg_uniform_superposition}). Then on the application of the $ X $ gate (ref.~line $ 4 $, Algorithm \ref{alg_uniform_superposition}) on $ \ket{q_{i}} $ for $ i = l_1 =2$ and $ i = l_2 = 3$ the following quantum state is obtained,
\begin{align}
     |  \smallunderbrace{1}_{ l_2 =3 }  \,  \smallunderbrace{1}_{ l_1=2 } \, 0 \,  \smallunderbrace{0}_{ l_0 =0 }  \rangle.
\end{align}
Then the application of the rotation $ R_Y(\theta_0) $ gate on $ \ket{q_{l_1}}$, with $  \theta_0 = - 2 \arccos \left(\sqrt {\frac{M_0}{M}}\right)  $, where $ M_0 = 2^{l_0} =1 $, gives the quantum state
\begin{align}
    b_0 |  \smallunderbrace{1}_{ l_2 =3 }  \,  \smallunderbrace{1}_{ l_1=2 } \, 0 \,  \smallunderbrace{0}_{ l_0 =0 }  \rangle + a_0 |  \smallunderbrace{1}_{ l_2 =3 }  \,  \smallunderbrace{0}_{ l_1=2 } \, 0 \,  \smallunderbrace{0}_{ l_0 =0 }  \rangle,
\end{align}
where 
\begin{align}
b_0 = \sqrt {\frac{2^{l_0}}{M}} = \sqrt {\frac{1}{13}},  \quad \text{and} \quad a_0 = \sqrt {\frac{M- 2^{l_0}}{M}} = \sqrt {\frac{12}{13}},	
\end{align}
(ref.~line $ 8 $, Algorithm \ref{alg_uniform_superposition}).
Subsequently, the application of the controlled Hadamard gate on $ \ket{q_{i}} $ for $ i = l_0 =0$,  to  $i= l_1 -1 =1$ conditioned on $  q_{l_1} = q_2$ being equal to $ 0 $ gives the quantum state,
\begin{align}
    b_0 |  \smallunderbrace{1}_{ l_2 =3 }  \,  \smallunderbrace{1}_{ l_1=2 } \, 0 \,  \smallunderbrace{0}_{ l_0 =0 }  \rangle + a_0 |  \smallunderbrace{1}_{ l_2 =3 }  \,  \smallunderbrace{0}_{ l_1=2 } \, + \,  \smallunderbrace{+}_{ l_0 =0 }  \rangle,
\end{align}
where $ \ket{+} = \frac{1}{\sqrt{2}} \left( \ket{0} +  \ket{1} \right) $ (ref.~line $ 9 $, Algorithm \ref{alg_uniform_superposition}).
Next we consider the ``For Loop'' in lines $ 10 $-$ 13 $ in Algorithm \ref{alg_uniform_superposition}. We note that in the first iteration, i.e., for $ m=1 $, the  
application of a controlled rotation on $ \ket{{q}_{l_{2}}} = \ket{q_3} $  conditioned on $  q_{l_1} = q_2$ being $ 0 $ (ref.~line $ 11 $, Algorithm \ref{alg_uniform_superposition}) results in the quantum state
\begin{align}
    b_0 |  \smallunderbrace{1}_{ l_2 =3 }  \,  \smallunderbrace{1}_{ l_1=2 } \, 0 \,  \smallunderbrace{0}_{ l_0 =0 }  \rangle + a_0 b_1 |  \smallunderbrace{1}_{ l_2 =3 }  \,  \smallunderbrace{0}_{ l_1=2 } \, + \,  \smallunderbrace{+}_{ l_0 =0 }  \rangle
     + a_0 a_1 |  \smallunderbrace{0}_{ l_2 =3 }  \,  \smallunderbrace{0}_{ l_1=2 } \, + \,  \smallunderbrace{+}_{ l_0 =0 }  \rangle.
\end{align}
Then the application of a controlled Hadamard ($ H $) gate on $ \ket{q_{i}} $ for $ i = l_1 =2$ to  $ i = l_2 -1 =2  $ conditioned on $ \ket{{q}_{l_{2}}} = \ket{{q}_{3}} $  being equal to $ 0 $, results in the quantum state 
\begin{align}
  \ket{\Psi} =  b_0 |  \smallunderbrace{1}_{ l_2 =3 }  \,  \smallunderbrace{1}_{ l_1=2 } \, 0 \,  \smallunderbrace{0}_{ l_0 =0 }  \rangle + a_0 b_1 |  \smallunderbrace{1}_{ l_2 =3 }  \,  \smallunderbrace{0}_{ l_1=2 } \, + \,  \smallunderbrace{+}_{ l_0 =0 }  \rangle
     + a_0 a_1 |  \smallunderbrace{0}_{ l_2 =3 }  \,  \smallunderbrace{+}_{ l_1=2 } \, + \,  \smallunderbrace{+}_{ l_0 =0 }  \rangle.
\end{align}
Here, 
\begin{align}
	b_1 = \sqrt {\frac{2^{l_1}}{M - 2^{l_0}}} = \sqrt {\frac{4}{12}},  \quad \text{and} \quad     a_1 = \sqrt {\frac{M- 2^{l_0} - 2^{l_1} }{M-2^{l_0} }} = \sqrt {\frac{8 }{12}}.  	
\end{align}
It follows that
			\begin{align} \label{eq_usp_state_vector_ex_one}
					\ket{\Psi} & = \sqrt {\frac{1}{13}} \, |  \smallunderbrace{1}_{ l_2 =3 }  \,  \smallunderbrace{1}_{ l_1=2 } \, 0 \,  \smallunderbrace{0}_{ l_0 =0 }  \rangle + \sqrt {\frac{4}{13}}\, |  \smallunderbrace{1}_{ l_2 =3 }  \,  \smallunderbrace{0}_{ l_1=2 } \, + \,  \smallunderbrace{+}_{ l_0 =0 }  \rangle
					+ \sqrt {\frac{8}{13}} \, |  \smallunderbrace{0}_{ l_2 =3 }  \,  \smallunderbrace{+}_{ l_1=2 } \, + \,  \smallunderbrace{+}_{ l_0 =0 }  \rangle \nonumber \\ & = 
					\frac{1}{\sqrt{13}} \sum_{j=0}^{12} \, \ket{j}.
				\end{align}
    This result was verified using IBM's Qiskit simulation environment. A histogram of sampling probabilities of obtaining various computational basis states is presented on the right side of \mfig{fig:uniform_ex_one}.
\end{example}
	
		\begin{example} \label{ex:usp_examples_ex_two}
		We illustrate how Algorithm \ref{alg_uniform_superposition} works by considering the case of $ M =13 \times 8 = 104$. Here $ M = 2^3 + 2^5 + 2^6$. Therefore, $l_0 = 3$, $l_1 = 5$ and $l_2 = 6$. 
	 	We will use only $n = \ceil{\log_2~ M} = 7 $ qubits to create the uniform superposition state  $\frac{1}{\sqrt{104}} \sum_{j=0}^{103} \, \ket{j}$.
		In this case, the quantum circuit produced by Algorithm \ref{alg_uniform_superposition} is shown in \mfig{fig:uniform_ex_two}. 
		The steps of Algorithm \ref{alg_uniform_superposition} are described in the following for this case.
  
\begin{figure}[H]
  \begin{center}
       \includegraphics[width=0.47\textwidth]{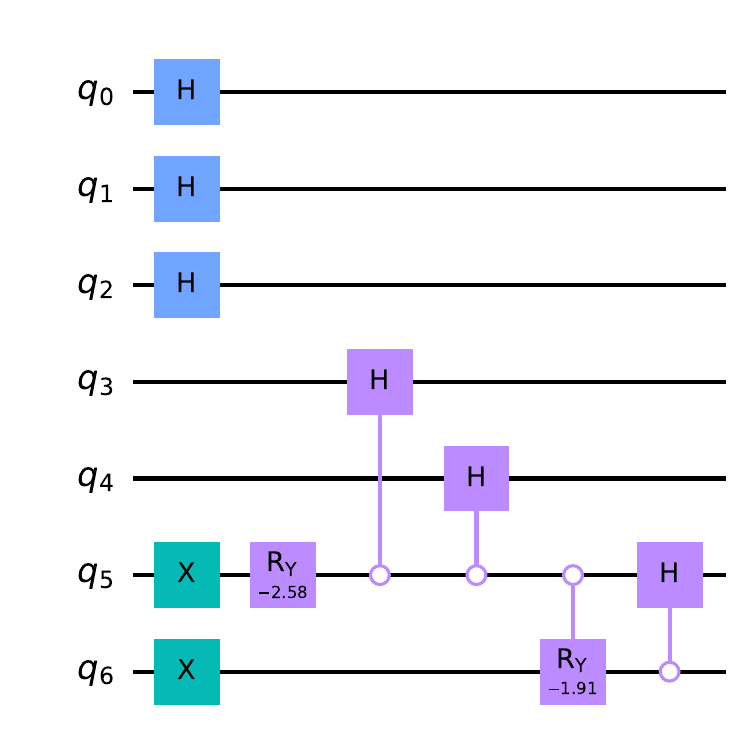}
  \end{center}
	 	\caption{A quantum circuit for creating the uniform quantum superposition state given in \meqref{eq_usp_state_vector_ex_two} is shown above. This corresponds to the case $ M=104 $ in Example \ref{ex:usp_examples_ex_two}. 
   }
	 	\label{fig:uniform_ex_two}
	 \end{figure}
  
   Each qubit is initialized to $\ket{0}$ (ref.~line $ 3 $, Algorithm \ref{alg_uniform_superposition}). Next, on the application of the $ X $ gate (ref.~line $ 4 $, Algorithm \ref{alg_uniform_superposition}) on $ \ket{q_{i}} $ for $ i = l_1 =5$ and $ i = l_2 = 6$ the following quantum state is obtained,
			\begin{align}
				|  \smallunderbrace{1}_{ l_2 =6 }  \,  \smallunderbrace{1}_{ l_1=5 } \, 0 \,  \smallunderbrace{0}_{ l_0 =3 } \, 0 \, 0 \, 0 \, \rangle.
			\end{align}
		  Since $ l_0  = 3 > 0 $, the application of the Hadamard gate on  $ \ket{q_{i}} $ for $ i = 0$, $ 1$,  and $ 2 $  results in the following quantum state 	(ref.~line $ 7 $, Algorithm \ref{alg_uniform_superposition}),
		  \begin{align}
		  	|  \smallunderbrace{1}_{ l_2 =6 }  \,  \smallunderbrace{1}_{ l_1=5 } \, 0 \,  \smallunderbrace{0}_{ l_0 =3 } \, + \, + \, + \, \rangle.
		  \end{align}
			Then the application of the rotation $ R_Y(\theta_0) $ gate on $ \ket{q_{l_1}} = \ket{q_{5}}$, with $  \theta_0 = - 2 \arccos \left(\sqrt {\frac{M_0}{M}}\right)  $, where $ M_0 = 2^{l_0} = 8 $, gives the quantum state
			\begin{align}
				b_0 \, |  \smallunderbrace{1}_{ l_2 =6 }  \,  \smallunderbrace{1}_{ l_1=5 } \, 0 \,  \smallunderbrace{0}_{ l_0 =3 } \, + \, + \, + \,  \rangle + a_0 \, |  \smallunderbrace{1}_{ l_2 =6 }  \,  \smallunderbrace{0}_{ l_1=5 } \, 0 \,  \smallunderbrace{0}_{ l_0 =3 } \, + \, + \, + \,  \rangle,
			\end{align}
			where 
			\begin{align}
				b_0 = \sqrt {\frac{2^{l_0}}{M}} = \sqrt {\frac{8}{104}},  \quad \text{and} \quad a_0 = \sqrt {\frac{M- 2^{l_0}}{M}} = \sqrt {\frac{96}{104}},	
			\end{align}
			(ref.~line $ 8 $, Algorithm \ref{alg_uniform_superposition}).
			Subsequently, the application of the controlled Hadamard gate on $ \ket{q_{i}} $ for $ i = l_0 =3$  to  $i= l_1 -1 =4$ conditioned on $  q_{l_1} = q_5$ being equal to $ 0 $ gives the quantum state,
			\begin{align}
				b_0 \, |  \smallunderbrace{1}_{ l_2 =6 }  \,  \smallunderbrace{1}_{ l_1=5 } \, 0 \,  \smallunderbrace{0}_{ l_0 =3 } \, + \, + \, + \, \rangle +  a_0 \, |  \smallunderbrace{1}_{ l_2 =6 }  \,  \smallunderbrace{0}_{ l_1=5 } \, + \,  \smallunderbrace{+}_{ l_0 =3 } \, + \, + \, + \, \rangle,
			\end{align}
			where $ \ket{+} = \frac{1}{\sqrt{2}} \left( \ket{0} +  \ket{1} \right) $ (ref.~line $ 9 $, Algorithm \ref{alg_uniform_superposition}).
			Next we consider the ``For Loop'' in lines $ 10 $-$ 13 $ in Algorithm \ref{alg_uniform_superposition}. We note that in the first iteration, i.e., for $ m=1 $, the  
			application of a controlled rotation on $ \ket{{q}_{l_{2}}} = \ket{q_6} $  conditioned on $  q_{l_1} = q_5$ being $ 0 $ (ref.~line $ 11 $, Algorithm \ref{alg_uniform_superposition}) results in the quantum state
			\begin{align}
				b_0 \,  |  \smallunderbrace{1}_{ l_2 =6 }  \,  \smallunderbrace{1}_{ l_1=5 } \, 0 \,  \smallunderbrace{0}_{ l_0 =3 } \, + \, + \, + \,  \rangle + a_0 b_1 |  \smallunderbrace{1}_{ l_2 =6 }  \,  \smallunderbrace{0}_{ l_1=5 } \, + \,  \smallunderbrace{+}_{ l_0 =3 } \, + \, + \, + \,  \rangle
				+ a_0 a_1 \, |  \smallunderbrace{0}_{ l_2 =6 }  \,  \smallunderbrace{0}_{ l_1=5 } \, + \,  \smallunderbrace{+}_{ l_0 =3 }  \, + \, + \, + \,  \rangle.
			\end{align}
			Then the application of a controlled Hadamard ($ H $) gate on $ \ket{q_{i}} $ for $ i = l_1 =5$ to  $ i = l_2 -1 =5  $ conditioned on $ \ket{{q}_{l_{2}}} = \ket{{q}_{6}} $  being equal to $ 0 $, results in the quantum state 
			\begin{align}
				\ket{\Psi} =  b_0 \, |  \smallunderbrace{1}_{ l_2 =6 }  \,  \smallunderbrace{1}_{ l_1=5 } \, 0 \,  \smallunderbrace{0}_{ l_0 =3 } \, + \, + \, + \,  \rangle + a_0 b_1 |  \smallunderbrace{1}_{ l_2 =6 }  \,  \smallunderbrace{0}_{ l_1=5 } \, + \,  \smallunderbrace{+}_{ l_0 =3 } \, + \, + \, + \,  \rangle
				+ a_0 a_1 \,  |  \smallunderbrace{0}_{ l_2 =6 }  \,  \smallunderbrace{+}_{ l_1=5 } \, + \,  \smallunderbrace{+}_{ l_0 =3 } \, + \, + \, + \, \rangle.
			\end{align}
			Here, 
			\begin{align}
				b_1 = \sqrt {\frac{2^{l_1}}{M - 2^{l_0}}} = \sqrt {\frac{32}{96}},  \quad \text{and} \quad     a_1 = \sqrt {\frac{M- 2^{l_0} - 2^{l_1} }{M-2^{l_0} }} = \sqrt {\frac{64 }{96}}.  	
			\end{align}
			It follows that
			\begin{align} \label{eq_usp_state_vector_ex_two}
				\ket{\Psi} & = \sqrt {\frac{8}{104}} \, |  \smallunderbrace{1}_{ l_2 =6 }  \,  \smallunderbrace{1}_{ l_1=5 } \, 0 \,  \smallunderbrace{0}_{ l_0 =3 } \, + \, + \, + \, \rangle + \sqrt {\frac{32}{104}}\, |  \smallunderbrace{1}_{ l_2 =6 }  \,  \smallunderbrace{0}_{ l_1=5 } \, + \,  \smallunderbrace{+}_{ l_0 =3 } \, + \, + \, + \, \rangle
				+ \sqrt {\frac{64}{104}} \, |  \smallunderbrace{0}_{ l_2 =6 }  \,  \smallunderbrace{+}_{ l_1=5 } \, + \,  \smallunderbrace{+}_{ l_0 =3 } \, + \, + \, + \, \rangle \nonumber \\ & = 
				 \frac{1}{\sqrt{104}} \sum_{j=0}^{103} \, \ket{j}.
			\end{align}
       The above result was verified using IBM's Qiskit simulation environment.  
		\end{example}

\begin{figure}[H]
\begin{center}
\includegraphics[width=0.49\textwidth]{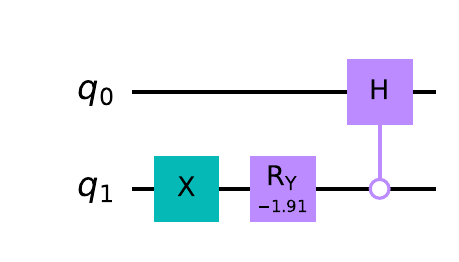}
\includegraphics[width=0.49\textwidth]{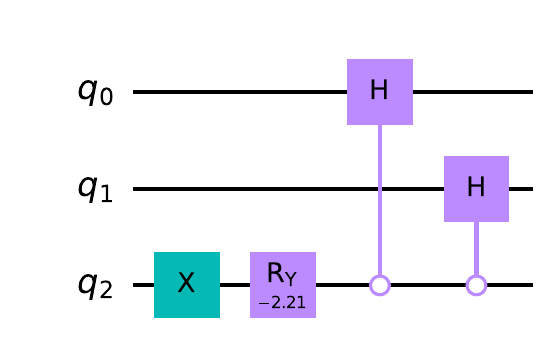} \\
\includegraphics[width=0.49\textwidth]{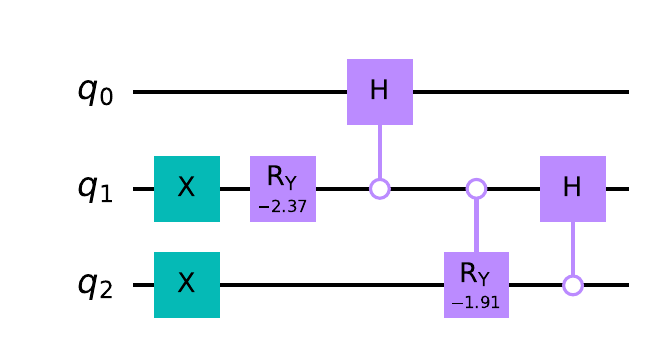}
\includegraphics[width=0.49\textwidth]{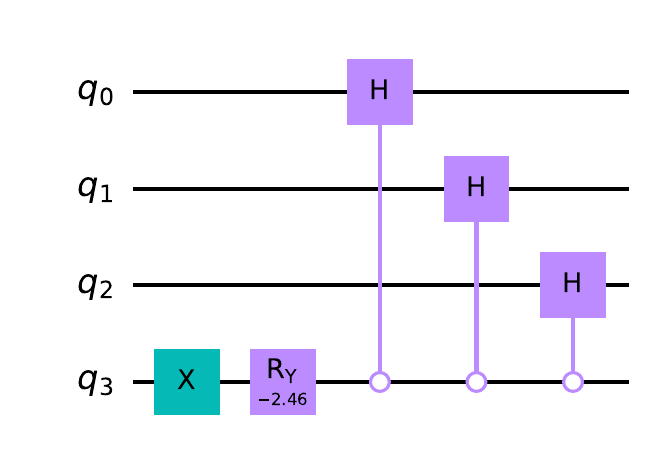}
\end{center}
\caption{Quantum circuits to obtain the uniform superposition states $\ket{\Psi} = \frac{1}{\sqrt{M}}\sum_{j=0}^{M-1} \ket{j}$, for (odd numbers) $M = 3$ (top, left), $5$ (top, right), $7$ (bottom, left) and $9$ (bottom, right) using Algorithm \ref{alg_uniform_superposition}.}
\label{fig:usp_odd}
\end{figure}

Quantum circuits for preparation of the uniform superposition states $\ket{\Psi} = \frac{1}{\sqrt{M}}\sum_{j=0}^{M-1} \ket{j}$ (based on Algorithm~\ref{alg_uniform_superposition}) are shown for selected cases in \mfig{fig:usp_odd} and \mfig{fig:usp_even}, where the number $M$ of distinct basis states in superposition is odd and even, respectively.
Algorithm \ref{alg_uniform_superposition} offers a highly efficient approach for creating the uniform superposition state $\ket{\Psi} = \frac{1}{\sqrt{M}}\sum_{j = 0}^{M - 1} \ket{j}$  by using only $\ceil{\log_2~ M}$ qubits.

We observe that there are additional Hadamard gates (at the top of the quantum circuit) for the even number cases in \mfig{fig:usp_even}. For instance, $M=6$ case in \mfig{fig:usp_even} contains an extra Hadamard gate in comparison to $M=3$ (in \mfig{fig:usp_odd}) case. For each factor of $2$ contained in $M$ there is a Hadamard gate in the circuit. For instance, the circuit for $M=12$ (in \mfig{fig:usp_even}) is similar to $M=3$ (in \mfig{fig:usp_odd}) except for 2 Hadamard gates at the top.
These observations can be related to line $ 7 $, Algorithm \ref{alg_uniform_superposition}, as $ l_0  > 0 $ when $M$ is even.

\begin{figure}[H]
\begin{center}
\includegraphics[width=0.4\textwidth]{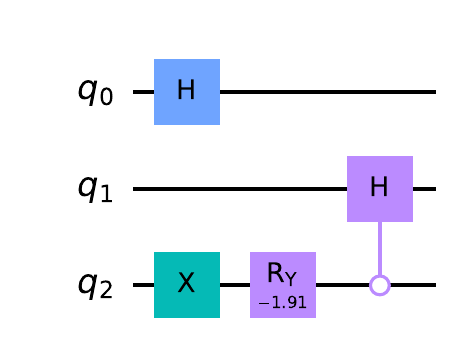}  \hspace{0.05\textwidth}
\includegraphics[width=0.4\textwidth]{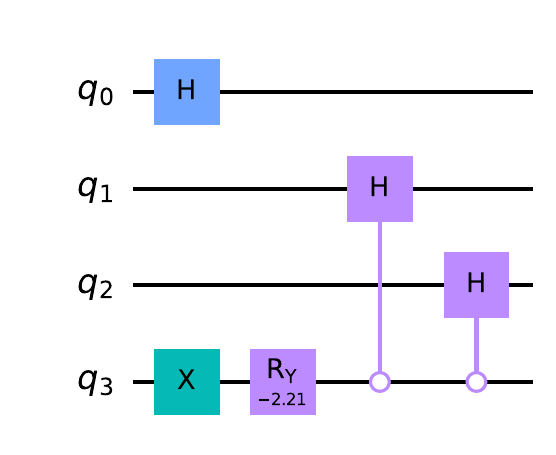}\\
\includegraphics[width=0.4\textwidth]{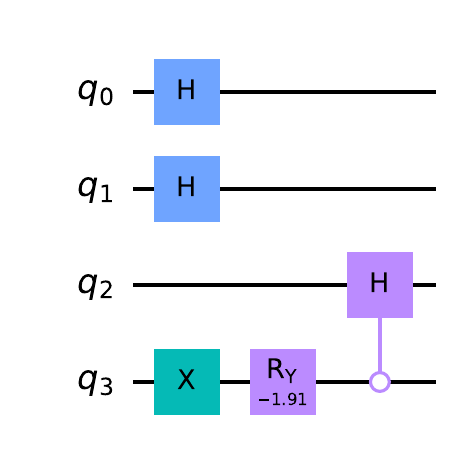} \hspace{0.05\textwidth}
\includegraphics[width=0.4\textwidth]{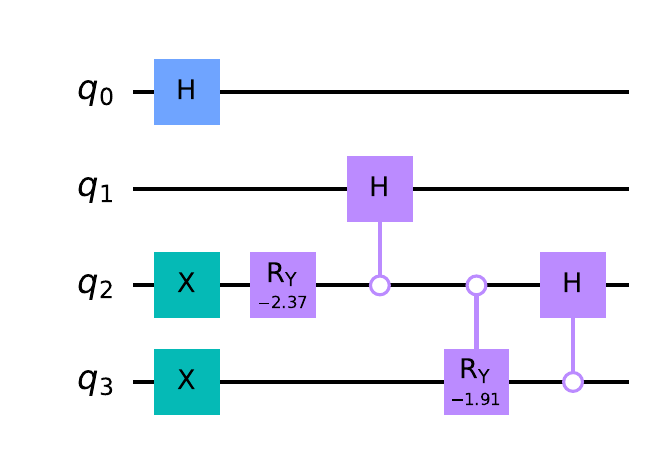}
\end{center}
\caption{Quantum circuits to obtain the uniform superposition states $\ket{\Psi} = \frac{1}{\sqrt{M}}\sum_{j=0}^{M-1} \ket{j}$, for (even numbers) $M = 6$ (top, left), $10$ (top, right), $12$ (bottom, left) and $14$ (bottom, right) using Algorithm \ref{alg_uniform_superposition}.}
\label{fig:usp_even}
\end{figure}

\subsection{Complexity analysis}
\label{sec:complexity_usp}

Let $ l_0$, $l_1$, $\ldots$, $l_k$, where $ M = \sum_{j=0}^{k} \, 2^{l_j} $ with $ 0 \leq l_0 < l_1 < \ldots < l_{k-1} < l_k \leq n-1 $. To obtain a uniform superposition of $M$ distinct states (where $ M \neq 2^r $ for any $ r \in \NN $) as  $\ket{\Psi} = \frac{1}{\sqrt{M}}  \,  \sum_{j=0}^{M-1} \, \ket{j}$, according to Algorithm \ref{alg_uniform_superposition}, we will need $l_k + 2k$ quantum gates (including 1 rotation ($R_Y(\theta)$) gate, $k$ Pauli-$X$ gates, $l_0$ Hadamard ($H$) gates (if $l_0 > 0$), $l_k - l_0$ controlled Hadamard  gates and $k-1$ controlled rotation $R_Y(\theta)$ gates. 
For the case where $M = 2^r$ for any $r \in \NN$, the uniform superposition state can be easily obtained using $k$ Hadamard gates. Hence, the number of elementary gates needed for creation of the uniform superposition state $\ket{\Psi_M} = \frac{1}{\sqrt{M}}  \,  \sum_{j=0}^{M-1} \, \ket{j}$ for any $M > 1$ and $M \in \NN$ is $O(n)$ or equivalently $O(\log_2 ~M)$. We note that the gate counts of each type (and the total number of gates) in the quantum circuits shown in \mfig{fig:usp_odd} and \mfig{fig:usp_even} are in agreement with the corresponding mathematical expressions given above. 

\begin{figure}[H]
\begin{center}
\includegraphics[width=0.9\textwidth]{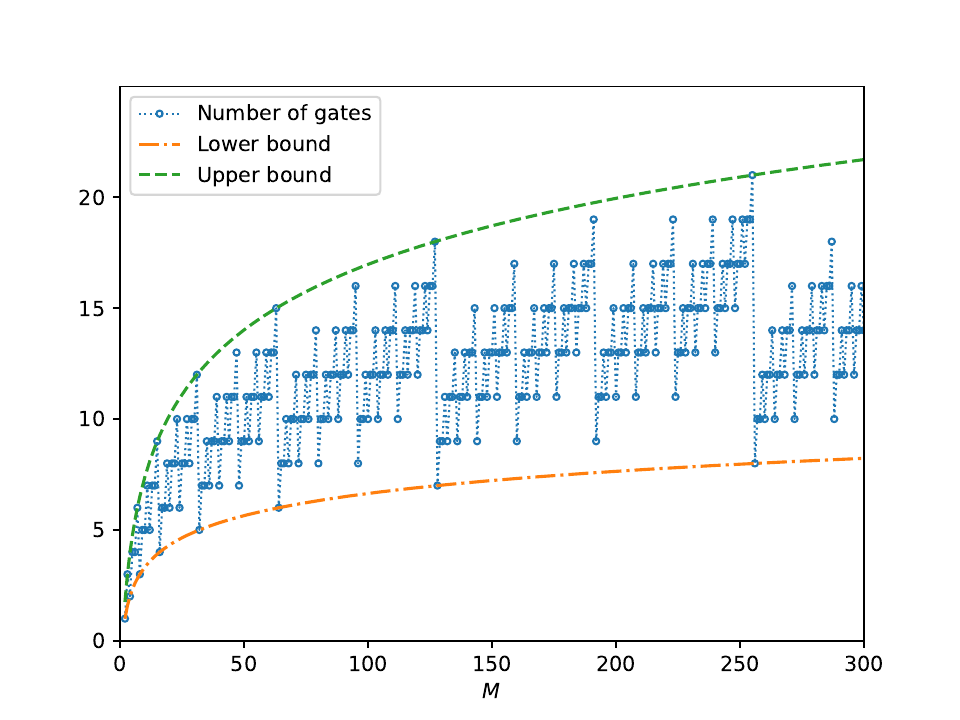}
\end{center}
\caption{Variation of the number of gates needed to obtain the uniform superposition state $\ket{\Psi} = \frac{1}{\sqrt{M}}\sum_{j=0}^{M-1} \ket{j}$ (according to Algorithm~\ref{alg_uniform_superposition}) with the number of distinct states $M$ in the uniform superposition state. Lower and upper bound curves are given by $\log_2~M$ (green dashes) and $3(\log_2 ~(M+1) - 1)$ (red dash-dots) respectively. }
\label{fig:ngates_algo1}
\end{figure}

 The dependence of the number of gates needed (i.e. $l_k + 2k$ gates as noted above) to obtain the uniform superposition state $\ket{\Psi} = \frac{1}{\sqrt{M}}\sum_{j=0}^{M-1} \ket{j}$ (according to Algorithm \ref{alg_uniform_superposition}) with the number of distinct states $M$ in the uniform superposition state is shown in \mfig{fig:ngates_algo1}. Bounds on the number of gates needed are also shown in this figure. The lower bound varies as $\log_2~M$ and is depicted by the red dash-dotted curve in \mfig{fig:ngates_algo1}. The upper bound varies as $3(\log_2~(M+1) - 1)$ as indicated by the green dashed curve. 
For any given $M$, the number of gates needed (according to Algorithm~\ref{alg_uniform_superposition} as indicated by the blue circles in \mfig{fig:ngates_algo1} is found to be bounded above and below by the green dashed curve and the red dash dotted curve respectively.
The number of gates needed to obtain the uniform superposition state $\ket{\Psi}= \frac{1}{\sqrt{M}}\sum_{j=0}^{M-1} \ket{j}$ is equal to the lower bound when $M=2^r$ for $r \in \NN$, as expected. A similar agreement between the number of gates needed and the upper bound occurs when $M = 2^r - 1$ for $r \in \NN$. 

Based on our proposed approach in Algorithm \ref{alg_uniform_superposition}, we observe that the preparation of the uniform superposition state $\ket{\Psi}$ can be equivalently obtained using only $O(\log_2 ~M)$ gates (including $O(\log_2 ~M)$ CNOT gates).
It follows from the fact that each of the controlled rotation gate and controlled Hadamard gates described in Algorithm \ref{alg_uniform_superposition} can be reconfigured using $O(1)$ CNOT gates. 
It can be verified from \mfig{fig:decompose}, which shows a method for constructing a controlled Hadamard gate using a single CNOT gate (at the top row), and the quantum circuit for constructing a controlled rotation gate ($R_Y(\theta)$) using two CNOT gates (at the bottom row). 

\begin{figure}[H]
  \begin{center}
       \includegraphics[width=0.65\textwidth]{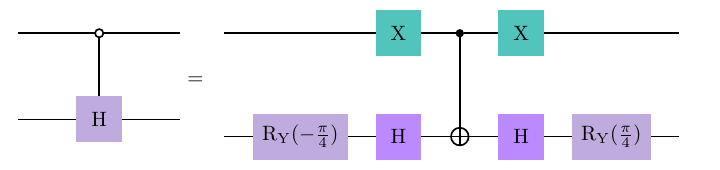} \\ \vspace{1cm}
      \includegraphics[width=0.65\textwidth]{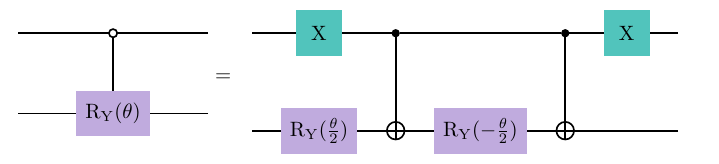}
  \end{center}
	 	\caption{
   The figure at the top illustrates the construction of a controlled Hadamard gate (with an open control, i.e., the Hadamard gate is applied on the target qubit if the control qubit is $\ket{0}$ and the identity operator is applied on the target qubit if the control qubit is $\ket{1}$.) using a CNOT gate and a few single-qubit gates. The figure at the bottom demonstrates the construction of a controlled rotation gate ($R_Y(\theta)$) (with an open control) using two CNOT gates and a few single-qubit gates.
   }
	 	\label{fig:decompose}
	 \end{figure}

\begin{table}[htbp]
  \caption{Comparison of the number of CNOT gates needed for the preparation of the uniform superposition states $\ket{\Psi} = \frac{1}{\sqrt{M}}  \,  \sum_{j=0}^{M-1} \, \ket{j} $ for several values of $M$ using our proposed approach and the state of the art implementation (Qiskit version 0.43.1). Table \ref{tab:quiskit_a}, \ref{tab:quiskit_b}, \ref{tab:quiskit_c} and \ref{tab:quiskit_d} show cases where $M$ is of the form $2^r-1$, $2^r + 2 $, $ 2^r + 1$, and $ 2^r - 2$, respectively, and in these cases, the number of CNOT gates needed using our approach is given by $3r-5$, $r-1$, $r$ and $3r-8$, respectively as shown in the third columns of the subtables.  Trends for the number of CNOT gates needed by Qiskit for these cases are shown in the fourth columns of the subtables. }  \label{tab:quiskit}
  \begin{subtable}[t]{0.5\linewidth}
     \centering
     \caption{Case: $M = 2^r -1$} \label{tab:quiskit_a}
    \begin{tabular}{cccc}
      \toprule
   r &  $M = 2^r -1 $ & $\substack{ \text{Proposed method} \\ \text{\#CNOTs  } = \, (3r-5)}$ & $\substack{\text{Qiskit} 
 \\ \text{\#CNOTs } = \, (2^r-2) }$\\
    \midrule
  2 &  3 & 1 & 2 \\
  3 &  7 & 4 & 6 \\
  4 &  15 & 7 & 14 \\
  5& 31 & 10 & 30 \\
  6 &  63 & 13 & 62 \\
  7&  127 & 16 & 126 \\
  8&  255 & 19 & 254 \\
   9& 511 & 22 & 510 \\
    10& 1023 & 25 & 1022 \\
    11& 2047 & 28 & 2046 \\
    12 & 4095 & 31 & 4094 \\
    13 & 8191 & 34 & 8190 \\
    14 & 16383 & 37 & 16382 \\
    15 & 32767 & 40 & 32766 \\
      \bottomrule
    \end{tabular}
  \end{subtable}%
  \begin{subtable}[t]{0.5\linewidth}
     \centering
   \caption{Case: $M = 2^r + 2$} \label{tab:quiskit_b}
    \raggedright{
    \begin{tabular}{cccc}
        \toprule
  $ r $ &  $M = 2^r +2  $ & $\substack{ \text{Proposed method} \\ \text{\#CNOTs  } = \, (r-1)}$ & $\substack{\text{Qiskit}   \\ \text{\#CNOTs } = \, (2^r + 2r-2) }$\\
    \midrule
 2&   6 & 1 & 6 \\
3&10 & 2 & 12 \\
4&18 & 3 & 22 \\
5&34 & 4 & 40 \\
6& 66 & 5 & 74 \\
7 &130 & 6 & 140 \\
8& 258 & 7 & 270 \\
9&514 & 8 & 528 \\
10&1026 & 9 & 1042 \\
11 & 2050 & 10 & 2068 \\
12& 4098 & 11 & 4118 \\
13 & 8194 & 12 & 8216 \\
14 & 16386 & 13 & 16410 \\
15 & 32770 & 14 & 32796 \\
      \bottomrule
    \end{tabular}}
  \end{subtable}
  
  \vspace{0.5cm}

  \begin{subtable}[t]{0.5\linewidth}
    \centering
   \caption{Case: $M = 2^r + 1$} \label{tab:quiskit_c}
    \begin{tabular}{cccc}
         \toprule
$ r $ &  $M = 2^r + 1  $ & $\substack{ \text{Proposed method} \\ \text{\#CNOTs  } = \, (r)}$ & $\substack{\text{Qiskit}   \\ \text{\#CNOTs } = \, (2r) }$\\
    \midrule
3 & 9 & 3 & 6 \\
4 & 17 & 4 & 8 \\
5 & 33 & 5 & 10 \\
6 & 65 & 6 & 12 \\
7 & 129 & 7 & 14 \\
8 & 257 & 8 & 16 \\
9 & 513 & 9 & 18 \\
10 & 1025 & 10 & 20 \\
11 & 2049 & 11 & 22 \\
12 & 4097 & 12 & 24 \\
13 & 8193 & 13 & 26 \\
14 & 16385 & 14 & 28 \\
15 & 32769 & 15 & 30 \\
      \bottomrule
    \end{tabular}
  \end{subtable}%
  \begin{subtable}[t]{0.5\linewidth}
   \caption{Case: $M = 2^r -2$} \label{tab:quiskit_d}
    \begin{tabular}{cccc}
        \toprule
   $ r $ &  $M = 2^r - 2  $ & $\substack{ \text{Proposed method} \\ \text{\#CNOTs  } = \, (3r -8)}$ & $\substack{\text{Qiskit}   \\ \text{\#CNOTs } = \, (2^r -2) }$\\
    \midrule
3& 6 & 1 & 6 \\
4 & 14 & 4 & 14 \\
5 & 30 & 7 & 30 \\
6 & 62 & 10 & 62 \\
7 &126 & 13 & 126 \\
8 &254 & 16 & 254 \\
9 & 510 & 19 & 510 \\
10 & 1022 & 22 & 1022 \\
11 & 2046 & 25 & 2046 \\
12 & 4094 & 28 & 4094 \\
13 & 8190 & 31 & 8190 \\
14 & 16382 & 34 & 16382 \\
15 & 32766 & 37 & 32766 \\
      \bottomrule
    \end{tabular}
  \end{subtable}
\end{table}

 Based on our proposed approach in Algorithm \ref{alg_uniform_superposition}, the fact that preparation of the uniform superposition state $\ket{\Psi} = \frac{1}{\sqrt{M}}\sum_{j=0}^{M-1} \ket{j}$  can be achieved using only $O(\log_2 ~M)$ gates (including $O(\log_2 ~M)$ CNOT gates)  
represents a very significant (exponential) reduction in gate complexity in comparison the state of the art implementation in Qiskit and Ref.~\cite{gleinig2021efficient} that requires $O(M)$ gates (including $O(M)$ CNOT gates). 
For instance, for $M=2^r-1$, the current Qiskit (Version 0.43.1) implementation (using Qiskit's transpile function) of the uniform superposition state $\ket{\Psi}$ is estimated to require $(2^r-2)$ CNOT gates, whereas our approach requires $(l_k - l_0)+2(k-1) = 3r -5$ CNOT gates (refer Table \ref{tab:quiskit_a}). 
For several other cases corresponding to different values of $M$, a comparison of the CNOT gate counts needed in our approach versus those required by Qiskit for the preparation of the uniform superposition states $\ket{\Psi} = \frac{1}{\sqrt{M}}  \,  \sum_{j=0}^{M-1} \, \ket{j} $  is shown in Table \ref{tab:quiskit}.
A graphical representation of the data provided in Table~\ref{tab:quiskit} is shown in 
\mfig{fig:mcases_all}. 
We note that for the cases shown in \ref{tab:quiskit_a}, \ref{tab:quiskit_b} and  \ref{tab:quiskit_d} our proposed approach offers an exponential reduction in the CNOT gate counts in comparison to the state of the art implementation in Qiskit. For the case shown in \ref{tab:quiskit_a}, the number of CNOT gates needed by
our approach is lower than the Qiskit implementation by a factor of $2$.  

A comparison of the number of CNOT gates required to prepare the uniform superposition states $\ket{\Psi}$ using our method and the Qiskit implementation for differnt values of $M$ (with $2 < M < 1024$ and $M \neq 2^r$ for any integer $r$) is presented in \mfig{fig:usp_cnot}.
It is clear from this comparison (presented in Table \ref{tab:quiskit}, \mfig{fig:mcases_all}, \mfig{fig:usp_cnot}) that our proposed approach achieves an exponential reduction in the number of CNOT gates needed by the Qiskit implementation in the general case. In a few cases, the Qiskit implementation does not depict an exponential increase in the number of CNOT gates with increasing $M$. These cases correspond to the extreme dips or valleys in the graph shown in \mfig{fig:usp_cnot} and occur when $M$ is of the form $M=2^r +1$ (also refer Table \ref{tab:quiskit_c}). In all cases, our proposed approach is superior to the corresponding Qiskit implementation.  

\begin{figure}[H]
\begin{center}
\includegraphics[width=\textwidth]{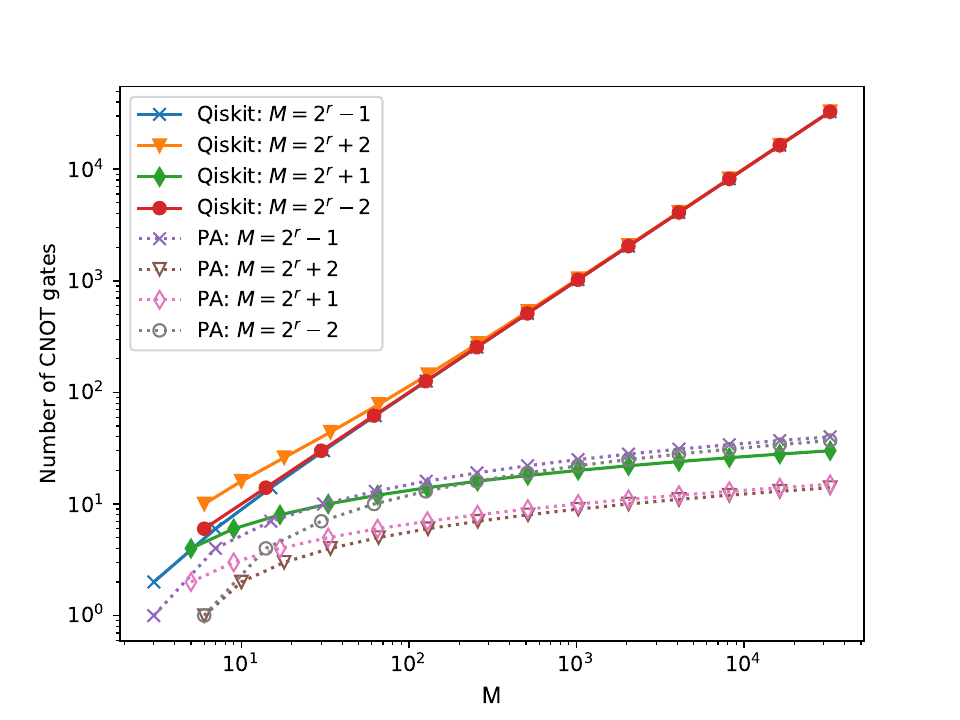}  
\end{center}
\caption{
Comparison of the number CNOT gates needed to obtain the uniform superposition states $\ket{\Psi} = \frac{1}{\sqrt{M}}\sum_{j=0}^{M-1} \ket{j}$, using our proposed approach (PA) versus Qiskit implementation for various cases of $M$ shown in Table~\ref{tab:quiskit}. It is evident that our proposed approach achieves an exponential improvement (in the number of CNOT gates) compared to the existing Qiskit implementation. 
}
\label{fig:mcases_all}
\end{figure}

\begin{figure}[H]
\begin{center}
\includegraphics[width=\textwidth]{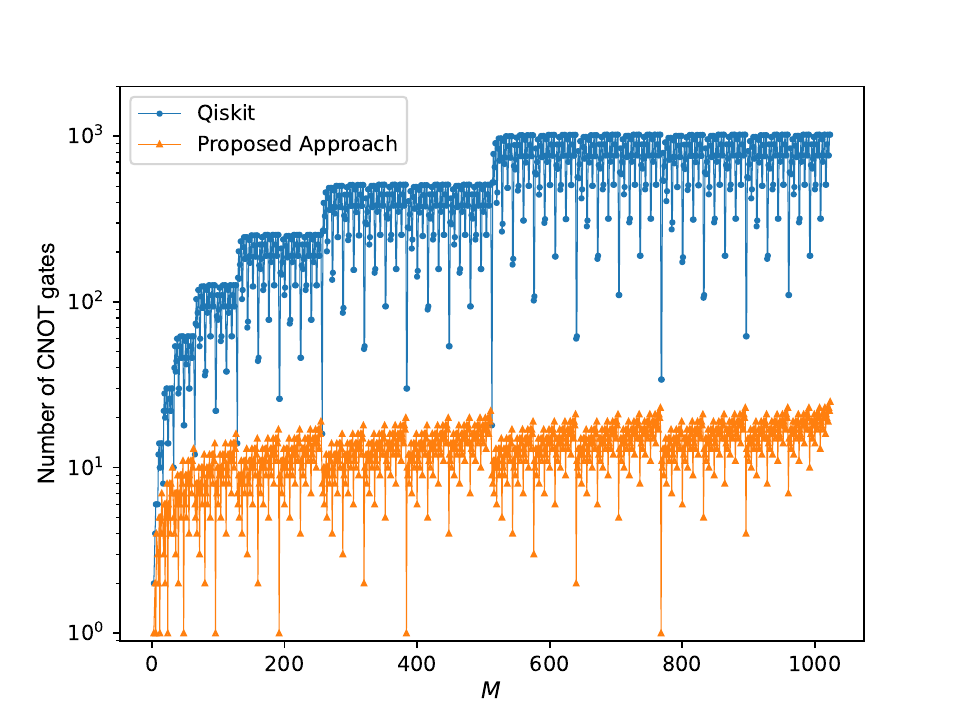}  
\end{center}
\caption{
Comparison of the number CNOT gates needed to obtain the uniform superposition states $\ket{\Psi} = \frac{1}{\sqrt{M}}\sum_{j=0}^{M-1} \ket{j}$, using our proposed approach versus Qiskit implementation for $  2 < M < 1024$. Note that $M=2^r$ cases for $r \in \NN$, are not shown in the figure as CNOT gates are not needed for these cases.}
\label{fig:usp_cnot}
\end{figure}

As noted earlier, the Quantum Byzantine Agreement (QBA) protocol~\cite{ben2005fast,mozafari2021efficient}, requires preparation of a uniform superposition state of the form
\begin{eqnarray} \label{eq:qba}
    \ket{\Phi} = \frac{1}{\sqrt{n^3}} \sum_{j=0}^{n^3 - 1} \ket{j} 
\end{eqnarray}
using $n$ qubits. This superposition state $\ket{\Phi}$ involves a uniform superposition over (the first) $M=n^3$ computational basis states out of a total of $2^n$ computational basis states.
Previous works in the literature~\cite{mozafari2021efficient} reported that preparation of such a state requires exponential CNOT gates in the worst case. 
Considering $n=20$ qubits, our approach for construction of a uniform superposition state $\ket{\Phi}$ shown in Eq.~\ref{eq:qba} requires 1 rotation gate, 5 Hadamard gates, $4$ controlled rotation gates and $6$ controlled Hadamard gates as illustrated in \mfig{fig:qba_n20}. 
Based on equivalence of controlled gates described in \mfig{fig:decompose}, we can infer that an equivalent circuit corresponding to \mfig{fig:qba_n20} would contain only 14 CNOT gates (along with a few single qubit gates). Similarly, for $n=18$, our approach based on Algorithm \ref{alg_uniform_superposition} needs $4$ controlled rotation gates and $9$ controlled Hadamard gates. This implies that only $17$ CNOT gates (along with a few single qubit gates) are needed for the construction of a uniform superposition state $\ket{\Phi}$ for $n=18$. This number is significantly (exponentially) lower than the estimates in Table~4 of Ref.~\cite{mozafari2021efficient}, where $2343$ CNOTs were needed for this case. 
\begin{figure}[H]
    \centering
\includegraphics[width=0.98\textwidth]{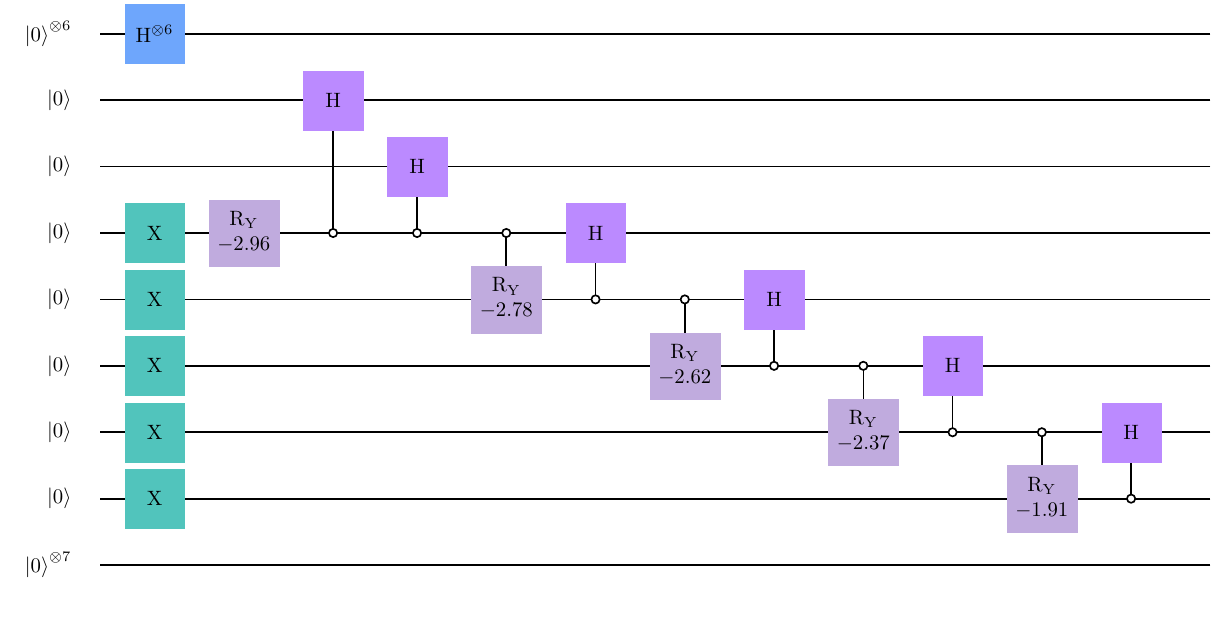}
    \caption{Quantum circuit to obtain the uniform superposition state, $    \ket{\Psi} = \frac{1}{\sqrt{n^3}} \sum_{j=0}^{n^3 - 1} \ket{j} $, with $n=20$ qubits (using Algorithm~\ref{alg_uniform_superposition}). The resulting state is relevant to the Quantum Byzantine Agreement (QBA) protocol.}
    \label{fig:qba_n20}
\end{figure}

\section{Nonuniform superposition}
\label{sec:nusp}

We note that in the quantum circuit created by Algorithm \ref{alg_uniform_superposition}, one rotation $R_Y(\theta)$ and gates $k-1$ controlled rotation gates were used. It is interesting to note that by changing the rotation angles for these gates many interesting quantum states can be created.

At the end of Algorithm \ref{alg_uniform_superposition}, with $ r=k $ in \meqref{eq_psi_algo_one_general}, the quantum state obtained is  
\begin{align} \label{eq_general_expression}
	\ket{\psi_{k-1}} = 
	\frac{b_0}{\sqrt{2^{l_0}}}  \sum_{j=0}^{2^{l_0} - 1} &\ket{j + M - 2^{l_0}} +  \frac{a_0 b_1 }{\sqrt{2^{l_1}}}  \sum_{j=0}^{2^{l_1} - 1} \ket{j + M - 2^{l_0} - 2^{l_1}}
	+ \frac{a_0 a_1 b_2 }{\sqrt{2^{l_2}}}   \sum_{j=0}^{2^{l_2} - 1} \ket{j + M - 2^{l_0} - 2^{l_1} - 2^{l_2} } \nonumber\\
	\cdots \cdots
	&+ \frac{a_0 a_1 \ldots a_{k-2} b_{k-1}}{\sqrt{2^{l_{k-1}}}} \sum_{j=0}^{2^{l_{k-1}} - 1} \ket{j + M - \sum_{s=0}^{k-1} \, 2^{l_s} }
	+ \frac{a_0 a_1 \ldots a_{k-2} a_{k-1} }{\sqrt{2^{l_{k}}}}\sum_{j=0}^{2^{l_{k}} - 1} \ket{  j + M - \sum_{s=0}^{k} \, 2^{l_s}}.
\end{align}
One can prepare the above more general quantum state by 
removing the restrictions on $\theta_0$ and $\theta_m$ (the rotation angles for the rotation and controlled rotation gates) in lines $8$ and $11$ in Algorithm \ref{alg_uniform_superposition}. Of course, the only constraint on the coefficient  $a_r$ and $b_r$ is the normalization requirement $\abs{a_r}^2 + \abs{b_r}^2 =1 $, for $r=0$ to $r=k-1$. In the following, the
quantum state $\ket{\psi_{k-1}}$ given in \meqref{eq_general_expression}, will be expressed in an alternate form and a few examples of such nonuniform superposition states will be considered.

Let \begin{align} \label{eq_short_def}
	\gamma_r  =
	\begin{cases}
		&  \frac{b_0}{\sqrt{2^{l_0}}} \quad \quad \quad \text{if } r=0, \\
		& \frac{a_0 a_1 \ldots a_{r-1} b_{r}}{\sqrt{2^{l_{r}}}} \quad \text{if }  0 <  r \leq k-1,  \\ 
		& \frac{a_0 a_1 \ldots a_{k-2} a_{k-1}} {\sqrt{2^{l_{k}}}} \quad \text{if }  r = k,
	\end{cases}
\end{align}
and
\begin{align} \label{eq_def_Gamma}
	\ket{\Gamma_r} = 	\sum_{j=0}^{2^{l_{r}} - 1} \ket{  j + M - \sum_{s=0}^{r} \, 2^{l_s}},
\end{align}
for $ r =0 $ to $ r=k $.
One can write \eqref{eq_general_expression} as \begin{align}
	\ket{\psi_{k-1}} = 	\sum_{r=0}^{k} \, \gamma_r \ket{\Gamma_r}.
\end{align}

In Algorithm \ref{alg_uniform_superposition}, the rotation angles for the rotation $R_Y(\theta)$ and controlled rotation gates were chosen such that the coefficients in the above expressions became equal, i.e., $ \gamma_i = \frac{1}{\sqrt{M}} $ for $ i=0 $ to $ i=k-1 $. In other words, 
\begin{align} \label{eq_coefficients_equal}
	\frac{ b_0 }{\sqrt{2^{l_0}}} = \frac{a_0 b_1 }{\sqrt{2^{l_1}}} =   \frac{a_0 a_1 b_2}{\sqrt{2^{l_2}}} = \frac{a_0 a_1 a_2 b_3}{\sqrt{2^{l_3}}}  = \cdots \cdots =  \frac{a_0 a_1 \ldots a_{k-2} b_{k-1}}{\sqrt{2^{l_{k-1}}}}  =  \frac{a_0 a_1 \ldots a_{k-2} a_{k-1} }{\sqrt{2^{l_{k}}}} = \frac{1}{\sqrt{M}}.
\end{align}
Therefore, the output of Algorithm \ref{alg_uniform_superposition} was a uniform superposition of $ M $ distinct states as desired. 
It is clear that if one or more of the above coefficients are made unequal (by changing the corresponding rotation angles for the rotation and control rotation gates), then one can obtain various combinations of nonuniform quantum states containing uniform quantum states as subsets of different sizes. In the following, some such examples will be considered.

\subsection{Example Circuits}	
\label{sec:nusp_examples}

\begin{example} \label{sec:nusp_examples_ex_one}
	Set $ b_i = a_i =  \frac{1}{\sqrt{2}} $ for $ i=0 $ to $ i = k-1 $, (or equivalently the rotation angle $ \theta = - \frac{\pi}{2} $ for all rotation and controlled rotation gates).  More precisely, in this case \meqref{eq_short_def} reduces to the following,
	\begin{align} \label{eq_short_def_example_one}
		\gamma_r  =
		\begin{cases}
			&  \frac{1}{\sqrt{2^{l_0+1}}} \quad \text{if } r=0, \\
			& \frac{1}{\sqrt{2^{l_{r} + r + 1 }}} \quad \text{if }  0 <  r \leq k-1,  \\ 
			& \frac{1} {\sqrt{2^{l_{k} +k}}} \quad \text{if }  r = k.
		\end{cases}
	\end{align}
Clearly, in this case, if $ i \neq j $, then $ \gamma_j \neq \gamma_j $, i.e., all the coefficients are distinct.
	Therefore, the quantum state obtained, as shown below, contains as subsets uniform
	 superposition of computational basis states of size $ 2^{l_0} $, $ 2^{l_1} $, $ \cdots $, $ 2^{l_{k-1}} $, and  $ 2^{l_{k}} $,
	 \begin{align}
	 	\sum_{r=0}^{k} \, \gamma_r \ket{\Gamma_r},
	 \end{align}
	 where  $ \ket{\Gamma_r} $ and $ \gamma_r $ are defined in \meqref{eq_def_Gamma} and  \meqref{eq_short_def_example_one}, respectively.

	 In the quantum circuit shown in \mfig{fig:nonuniform_ex_one}, the case $ M=15  $ is considered. As $ M = 15 = 2^0 + 2^1 + 2^2 + 2^3 $, in this case $ l_0 = 0 $, $ l_1 =1$, $ l_2=2 $ and $ l_3 =3 $, with $k=3$. In the quantum circuit in \mfig{fig:nonuniform_ex_one}, one $ R_Y (\theta) $ rotation gate and two controlled rotations gates are used with the rotation angle $ \theta = - \frac{\pi}{2} $ for each of these gates. Using \meqref{eq_short_def_example_one} one can obtain 
	 \begin{align}
	 	\gamma_0 = \frac{1}{\sqrt{2}}, \, \gamma_1 = \frac{1}{\sqrt{8}}, \, \gamma_2 = \frac{1}{\sqrt{32}} \, \text{ and }\, \gamma_3 = \frac{1}{\sqrt{64}}.
	 \end{align}
	 It follows from \meqref{eq_def_Gamma} that the output quantum state obtained  is   
	 \begin{align}\label{eq_state_vector_ex_one}
	 	\frac{1}{\sqrt{2}} \left( \ket{14} \right) +  	\frac{1}{\sqrt{8}} \left( \ket{12} + \ket{13} \right) + 	\frac{1}{\sqrt{32}} \left( \ket{8} + \ket{9} + \ket{10} + \ket{11} \right)  +  \frac{1}{\sqrt{64}} \left( \ket{0} + \ket{1} + \ket{2} + \ket{3} + \ket{4} + \ket{5} + \ket{6} + \ket{7}    \right). 
	 \end{align}
    The above was verified using IBM's Qiskit simulation environment. A histogram of sampling probabilities of obtaining various computational basis states is presented on the right side of \mfig{fig:nonuniform_ex_one}.
	 
	 \begin{figure}[H]
  \begin{center}
       \includegraphics[width=0.47\textwidth]{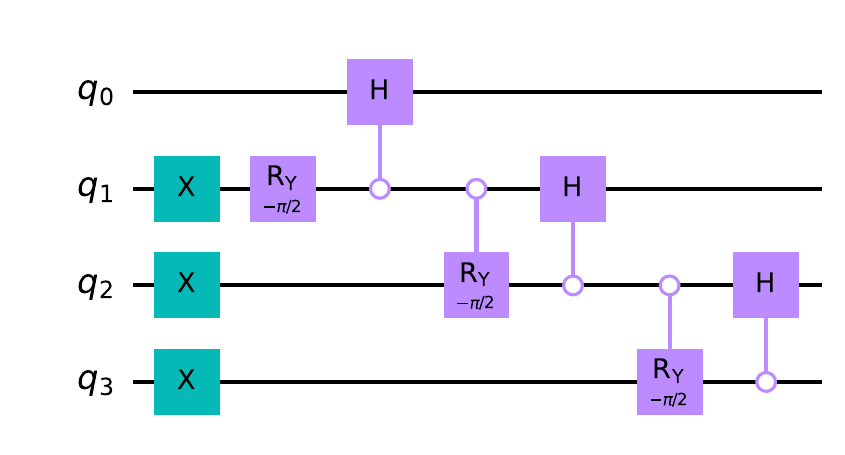}
        \includegraphics[width=0.52\textwidth]{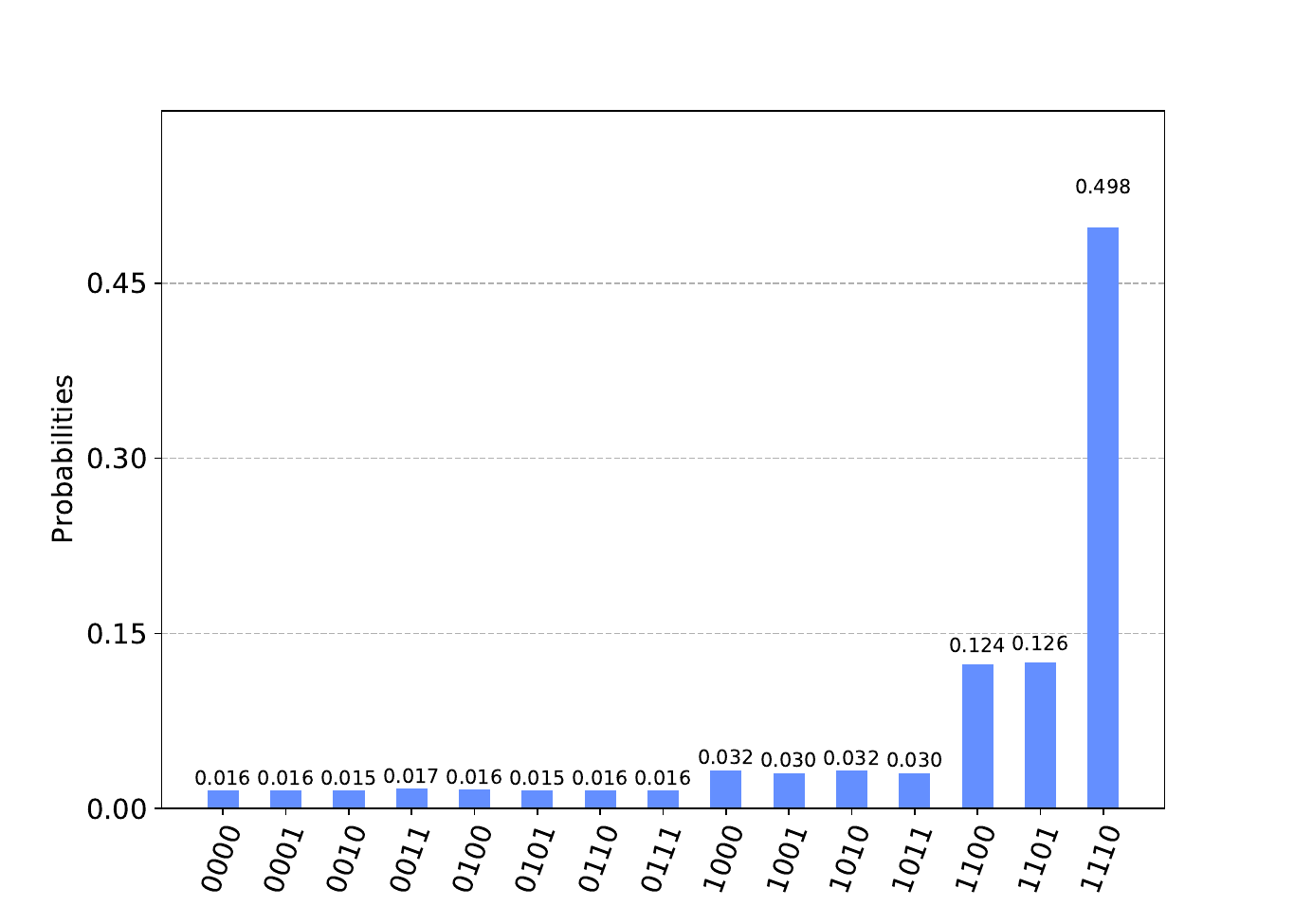}
  \end{center}
	 	\caption{A quantum circuit for creating the nonuniform quantum state given in \meqref{eq_state_vector_ex_one} is shown on the left. This corresponds to the case $ M=15 $ in Example \ref{sec:nusp_examples_ex_one}. A histogram representing the sampling probabilities of obtaining various computational basis states is shown on the right.}
	 	\label{fig:nonuniform_ex_one}
	 \end{figure}
	 
	\end{example}

\begin{example} \label{sec:nusp_examples_ex_two}
Suppose $a_s =0 $. For this case, the rotation angle $ \theta_s  =0 $ for the corresponding controlled rotation  $ R_Y(\theta_s) $ gate), where $ 0 < s \leq k-1 $. 
We also assume that all the other rotation angles in Algorithm \ref{alg_uniform_superposition} remain unchanged (i.e.,  $ \theta_r = - 2 \arccos \left(\sqrt {\frac{2^{l_r}}{M-\sum_{j=0}^{r-1} 2^{l_j} }}\right)  $ for  $ r \neq  s$). Since $ a_s $ is a factor of $ \gamma_j $ for $ j \geq s +1 $, it is clear that $ \gamma_j = 0 $ for $ j  \geq s +1 $. Also,  if  $ j < s $ then $ \gamma_j  $ remains the same as in  Algorithm \ref{alg_uniform_superposition}, i.e.,  $ \gamma_r = \frac{1}{\sqrt{M}}$ for $ 0 \leq r  < s $. Further, if $a_s =0 $, then $ b_s=1 $. Therefore, 
\begin{align}\label{eq_def_gamma_s_minus_one}
	\gamma_{s} = \frac{ a_0 a_1 \cdots a_{s-1}}{\sqrt{2^{l_s}}} = \sqrt{\frac{M- \sum_{j=0}^{s-1} 2^{l_j}}{M 2^{l_s}}}.
\end{align}
Therefore, in this case the following quantum state is obtained,
 \begin{align}
	\sum_{r=0}^{s} \, \gamma_r \ket{\Gamma_r} =  \frac{1}{\sqrt{M}} \left(\sum_{r=0}^{s-1} \,  \ket{\Gamma_r} \right) + 	\gamma_{s} \ket{\Gamma_{s}},
\end{align}
where  $ \ket{\Gamma_r} $ is defined in \meqref{eq_def_Gamma} and  $ \gamma_{s} $ is defined in \meqref{eq_def_gamma_s_minus_one}.

 In quantum circuit shown in \mfig{fig:nonuniform_ex_TWO}, the case of $ M=31 $ and $ s=2 $ is considered. As $ M = 31 = 2^0 + 2^1 + 2^2 + 2^3 +2^4 $, in this case $ l_0 = 0 $, $ l_1 =1$, $ l_2=2 $, $ l_3 =3 $, and $l_4 =4$ with $k=4$. Since, $ s =2 $, it means $ a_2 = 0  $ (or equivalently  $ \theta_2  =0 $ for the controlled rotation gate). It follows from the discussion above that
 \begin{align*}
 	\gamma_0 = \gamma_1 = \frac{1}{\sqrt{31}}, \gamma_2  = \sqrt{\frac{7}{31}}, \text{ and } \gamma_3 = \gamma_4 = 0,  
 \end{align*}
and the output quantum state obtained is
\begin{align} \label{eq_state_vector_ex_two}
	\frac{1}{\sqrt{31}} \left(\ket{30}\right)  + 	\frac{1}{\sqrt{31}} \left(\ket{28} + \ket{29}\right) + \sqrt{\frac{7}{31}} \left(\ket{24} + \ket{25} + \ket{26} + \ket{27} \right).
\end{align}
The quantum circuit depicted on the left side of \mfig{fig:nonuniform_ex_TWO} was created and executed within IBM's Qiskit simulation environment. 
The results obtained were confirmed to be correct. A histogram of sampling probabilities of obtaining various computational basis states is presented on the right side of \mfig{fig:nonuniform_ex_one}.
 
\begin{figure}[H]
  \begin{center}
      \includegraphics[width=0.47\textwidth]{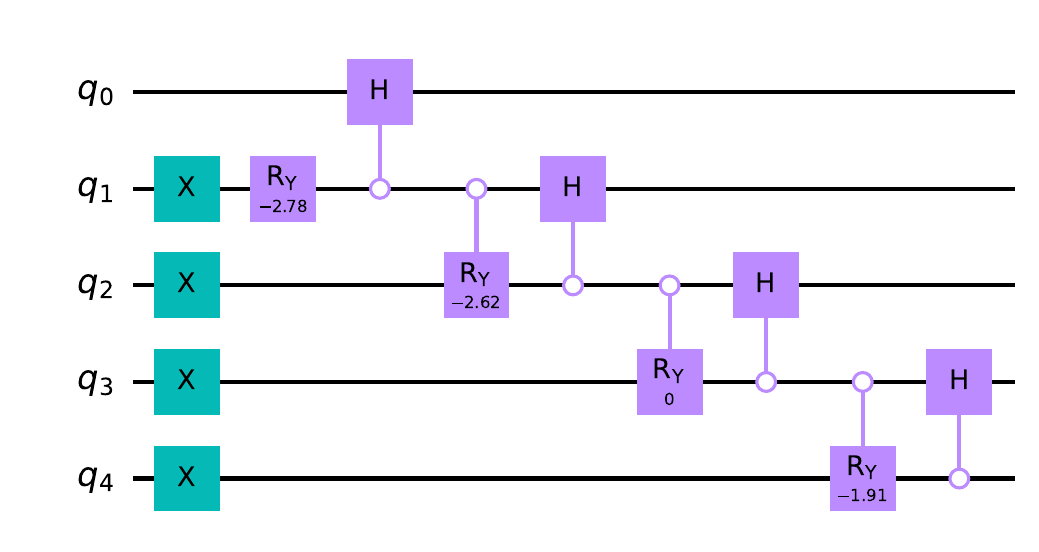} 
            \includegraphics[width=0.52\textwidth]{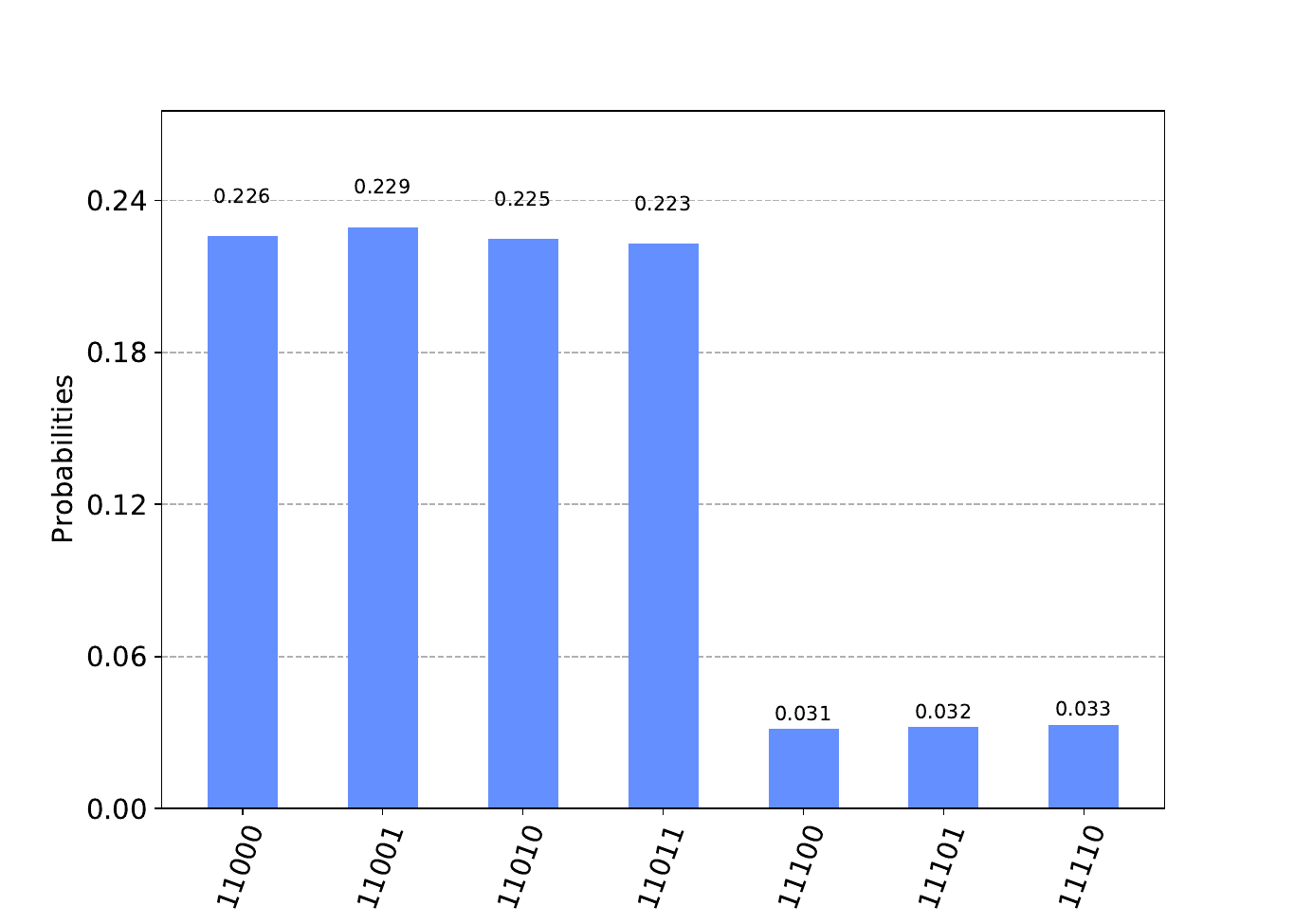}
  \end{center}
	\caption{A quantum circuit for creating the nonuniform quantum state given in \meqref{eq_state_vector_ex_two} is shown on the left. This corresponds to the case $ M=31 $ and $ s=2 $ in Example \ref{sec:nusp_examples_ex_two}. A histogram representing the sampling probabilities of obtaining various computational basis states is shown on the right.}
		\label{fig:nonuniform_ex_TWO}
\end{figure}
\end{example}

\begin{example}\label{sec:nusp_examples_ex_three}
	Suppose $b_s =0 $. For this case, the rotation angle $ \theta_s  = - \pi $ for the corresponding  controlled rotation gate $ R_Y (\theta_s) $, where  $ 0 < s \leq k-1 $. 
	Similar to the previous example, we also assume that all the other rotation angles in Algorithm \ref{alg_uniform_superposition} remain unchanged (i.e.,  $ \theta_r = - 2 \arccos \left(\sqrt {\frac{2^{l_r}}{M-\sum_{j=0}^{r-1} 2^{l_j} }}\right)  $ for  $ r \neq  s$).
	It is clear that in this case 	$ \gamma_s = 0 $, as $ b_s $ is a factor of $ \gamma_s $. 
A simple calculation shows that the following quantum state is obtained in this case,
\begin{align}
\frac{1}{\sqrt{M}}	\left(\sum_{r=0}^{s-1} \, \ket{\Gamma_r}\right) +  \sqrt {\frac{M -    \sum_{j=0}^{s-1} 2^{l_j} } {M \left(M -    \sum_{j=0}^{s} 2^{l_j} \right)}}  \left(	\sum_{r=s+1}^{k-1} \,  \ket{\Gamma_r} \right). 
\end{align}
An easy computation shows that, for $ M =15 $ and $ s =2 $, the quantum state obtained is
\begin{align} \label{eq_state_vector_ex_three}
	\frac{1}{\sqrt{15}}	\left( \ket{\Gamma_0}  + \ket{\Gamma_1} \right) +  	\frac{1}{\sqrt{10}}   (\ket{\Gamma_3} + \ket{\Gamma_4} ) 
		 =	\frac{1}{\sqrt{15}} \left( \ket{12} + \ket{13} + \ket{14} \right)  +  \frac{1}{\sqrt{10}} \left( \ket{0} + \ket{1} + \ket{2} + \ket{3} + \ket{4} + \ket{5} + \ket{6} + \ket{7}    \right). 
\end{align}
  The above was verified using the quantum circuit shown on the left side of \mfig{fig:nonuniform_ex_three} in IBM's Qiskit simulation environment. A histogram of sampling probabilities of obtaining various computational basis states is presented on the right side of \mfig{fig:nonuniform_ex_three}.
 
	\begin{figure}[H]
		\includegraphics[width=0.47\textwidth]{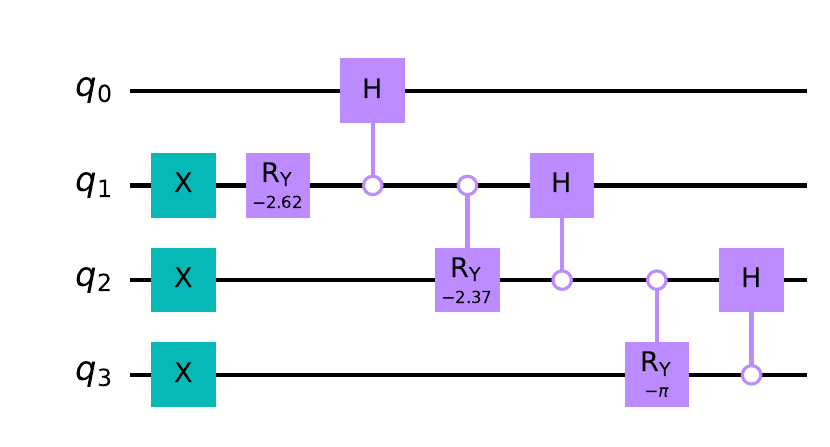} 
         \includegraphics[width=0.52\textwidth]{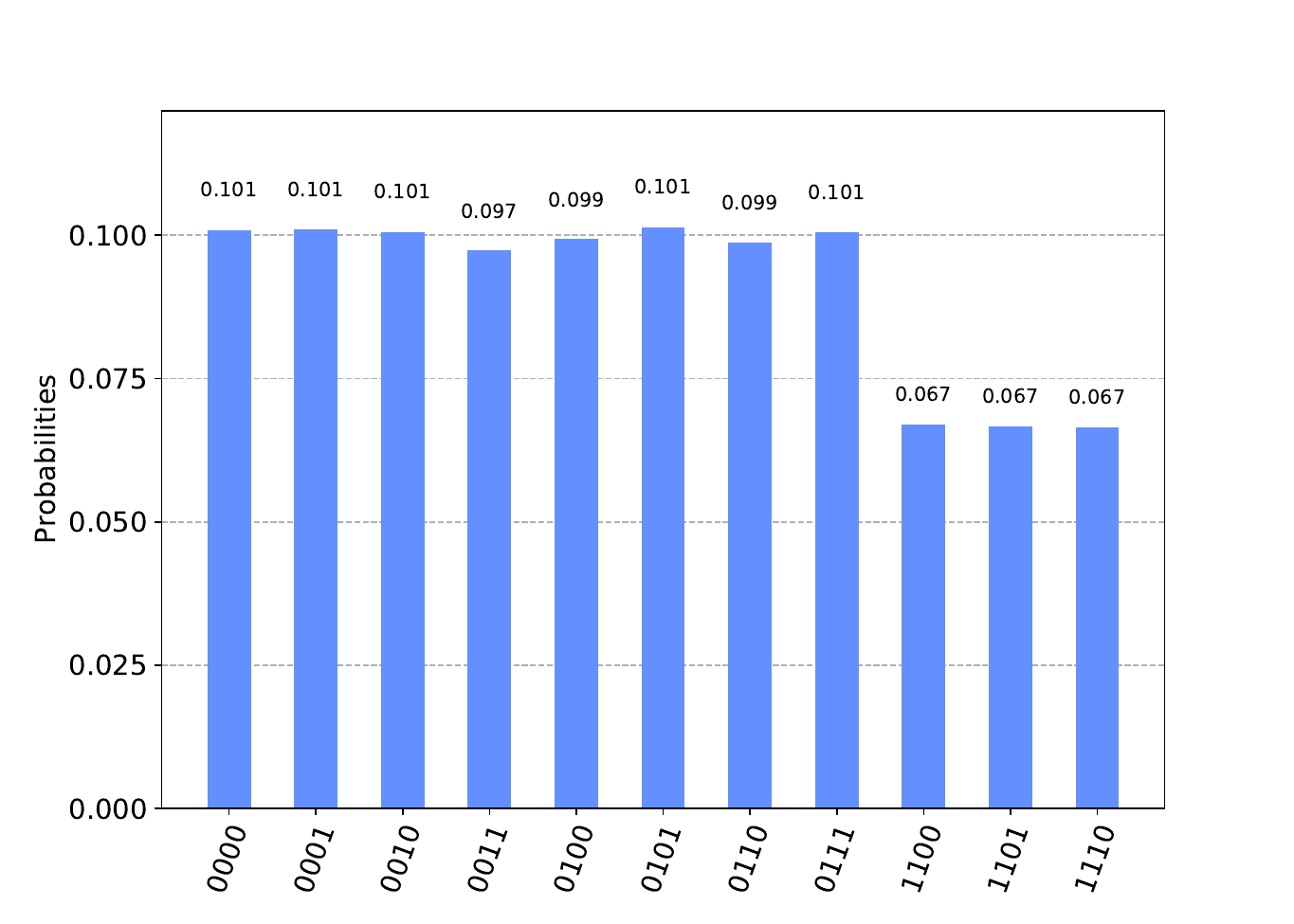}
		\caption{A quantum circuit for creating the nonuniform quantum state given in \meqref{eq_state_vector_ex_three} is shown on the left. This corresponds to the case $ M=15 $ and $ s=2 $ in Example \ref{sec:nusp_examples_ex_three}. A histogram representing the sampling probabilities of obtaining various computational basis states is shown on the right.}
		\label{fig:nonuniform_ex_three}
	\end{figure}	
\end{example}

\section{Conclusion}
\label{sec:conclusion}

In this paper, we proposed an efficient solution to the problem of quantum state preparation involving a uniform superposition over a non-empty subset of $n$-qubit computational basis states. 
The uniform superposition state considered was of the form $\ket{\Psi} = \frac{1}{\sqrt{M}}\sum_{j = 0}^{M - 1} \ket{j}$, where $M$ denotes the number of distinct states in the superposition state and $2 \leq M \leq 2^n$. 
%
We showed that this uniform superposition state  $\ket{\Psi} $,
can be created (based on Algorithm~\ref{alg_uniform_superposition}) using only $O(\log_2 ~M)$ elementary quantum gates. This represents a significant (exponential) reduction in gate complexity in comparison to previous works~\cite{gleinig2021efficient, mozafari2021efficient, Qiskit}. 
In addition to gate complexity, the circuit depth associated with creation of the uniform superposition state was also found to be $O(\log_2 ~M)$.
Further, only $n=\ceil{\log_2 ~M}$ qubits are needed for preparation of the uniform superposition state $\ket{\Psi}$  for arbitrary $M$. Note that no ancillary qubits are needed in our approach. 
Moreover, our approach (in Algorithm~\ref{alg_uniform_superposition}) does not require  controlled quantum gates with multiple controls.
Only appropriate combinations of single qubit gates (namely Pauli X gates, Hadamard gates, rotation ($R_Y (\theta )$) gates) and controlled gates with a single control  (namely controlled Hadamard gates and controlled rotation gates) are used.
Mathematical expressions for the number of gates of each type, total number of gates, along with lower and upper bounds are presented in \mref{sec:complexity_usp}. Note that the controlled Hadamard gates and controlled rotation gates can be implemented using CNOT gates and a few single qubit gates. Comparisons presented in Table \ref{tab:quiskit}, \mfig{fig:mcases_all} and \mfig{fig:usp_cnot} demonstrate that in the general case, our proposed approach achieves an exponential reduction in the number of CNOT gates compared to the existing Qiskit implementation.

Further, we showed (in \mref{sec:nusp}) that the same quantum circuit configuration used for creating uniform superposition state  $\ket{\Psi} $, described above, can also be used to create a broad class of nonuniform superposition states or mixed states. In such a class of nonuniform superposition states, multiple uniform superpositions are allowed to occur over different subsets of the computational basis states. In other words, for a given $M$, the same quantum circuit configuration (as the one used to generate the uniform superposition state $\ket{\Psi} $) can be used with appropriately modified rotation angles associated with rotation gates and controlled rotation gates to generate special partitions of $M$ quantum computational basis states into multiple subsets where the amplitudes are constant within each subset but can vary across subsets.  
Hence a broad class of nonuniform superposition states can also be efficiently prepared with a gate complexity and circuit depth of $O(\log_2~M)$ using only $n=\ceil{\log_2 ~M}$ qubits.

It is anticipated that our proposed approaches for the efficient preparation of uniform superposition states (and also selected nonuniform superposition states) over subsets of computational basis states will be useful in many applications in areas such as cryptography, error correcting codes, quantum solution of linear system of equations, quantum solution of differential equations and quantum machine learning, among others.

\section{Disclaimer}
Unfortunately, several false claims and misinformation have appeared in the media regarding the relationship between the authors of this work. Both authors deny the factually incorrect claims that the first author, Alok Shukla, is a former student of the second author, Prakash Vedula, or that Prakash Vedula led the research team contributing to this work. The authors' contributions have been clarified in the ``Author Contributions'' section at the beginning of the document.

	\bibliographystyle{unsrt}

\end{document}